\documentclass[11pt]{article}

\usepackage{caption}
\usepackage{algorithm}
\usepackage{scalerel}
\usepackage{mathrsfs}
\usepackage{enumitem}
\usepackage{bbm}
\usepackage{tikz}
\usepackage{float}
\usetikzlibrary{backgrounds}
\usepackage{color}
\usepackage{graphicx}
\usepackage{latexsym}
\usepackage{amsfonts}
\usepackage{pifont,xspace,fullpage,epsfig, wrapfig}
\usepackage{amsmath, amssymb, amsthm} 
\usepackage{multirow}
\usepackage{dsfont}
\usepackage{array}
\usepackage{adjustbox}

\usepackage{bookmark}

\DeclareSymbolFont{AMSb}{U}{msb}{m}{n}
\DeclareMathSymbol{\N}{\mathbin}{AMSb}{"4E}
\DeclareMathSymbol{\Z}{\mathbin}{AMSb}{"5A}
\DeclareMathSymbol{\R}{\mathbin}{AMSb}{"52}
\DeclareMathSymbol{\Q}{\mathbin}{AMSb}{"51}
\DeclareMathSymbol{\erert}{\mathbin}{AMSb}{"50}
\DeclareMathSymbol{\I}{\mathbin}{AMSb}{"49}
\DeclareMathSymbol{\C}{\mathbin}{AMSb}{"43}

\definecolor{gray}{gray}{0.4}

\newcommand{\remove}[1]{}

\newtheorem{theorem}{Theorem}[section]
\newtheorem{lemma}[theorem]{Lemma}
\newtheorem{definition}[theorem]{Definition}

\newtheorem{claim}[theorem]{Claim}

\newtheorem{observation}[theorem]{Observation}

\newtheorem{example}[theorem]{Example}

\newcommand{\1}{\mathbbm{1}}

\newcommand{\AAA}{\mathcal A}
\newcommand{\BBB}{\mathcal B}

\newcommand{\DDD}{\mathcal D}
\newcommand{\FFF}{\mathcal F}

\newcommand{\eps}{\varepsilon}

\newcommand{\e}{\mathrm{e}}

\def\Q{\operatorname*{\mathbb{Q}}}

\def\OPT{\mathop{\rm{OPT}}\nolimits}

\makeatletter
\newcommand{\thickhline}{
    \noalign {\ifnum 0=`}\fi \hrule height 1pt
    \futurelet \reserved@a \@xhline
}
\newcolumntype{"}{@{\hskip\tabcolsep\vrule width 1pt\hskip\tabcolsep}}
\makeatother

\makeatletter
\newlength{\fboxhsep}
\newlength{\fboxvsep}
\newlength{\fboxtoprule}
\newlength{\fboxbottomrule}
\newlength{\fboxleftrule}
\newlength{\fboxrightrule}
\def\@frameb@xother#1{%
  \@tempdima\fboxtoprule
  \advance\@tempdima\fboxvsep
  \advance\@tempdima\dp\@tempboxa
  \hbox{%
    \lower\@tempdima\hbox{%
      \vbox{%
        \hrule\@height\fboxtoprule
        \hbox{%
          \vrule\@width\fboxleftrule
          #1%
          \vbox{%
            \vskip\fboxvsep
            \box\@tempboxa
            \vskip\fboxvsep}%
          #1%
          \vrule\@width\fboxrightrule}%
        \hrule\@height\fboxbottomrule}%
    }%
  }%
}
\long\def\fboxother#1{%
  \leavevmode
  \setbox\@tempboxa\hbox{%
    \color@begingroup
    \kern\fboxhsep{#1}\kern\fboxhsep
    \color@endgroup}%
  \@frameb@xother\relax}

\makeatother

\begin{document}

\begin{titlepage}

\title{Privately Learning Thresholds: Closing the Exponential Gap}

\author{
Haim Kaplan\thanks{Tel Aviv University and Google Research. \tt haimk@post.tau.ac.il}
\and
Katrina Ligett\thanks{School of Computer Science and Engineering, Hebrew University of Jerusalem, Jerusalem 91904, Israel. Email:
  \texttt{katrina@cs.huji.ac.il}. NSF grants CNS-1254169 and CNS-1518941, US-Israel Binational Science Foundation grant 2012348, Israel Science Foundation (ISF) grant \#1044/16, United States Air Force and DARPA under contract FA8750-16-C-0022, and the Federmann Cyber Security Center in conjunction with the Israel national cyber directorate. Any opinions,
findings and conclusions or recommendations expressed in this material
are those of the author(s) and do not necessarily reflect the views of
the United States Air Force and DARPA.
}
\and
Yishay Mansour\thanks{Tel Aviv University and Google Research. \tt mansour.yishay@gmail.com}
\and
Moni Naor\thanks{Department of Computer Science and Applied Mathematics,
  Weizmann Institute of Science,  Rehovot 76100, Israel. Email:
  \texttt{moni.naor@weizmann.ac.il}. Supported in part by grant  from the Israel
  Science Foundation (no.\ 950/16) and the US-Israel Binational Science Foundation grant 2012348. Incumbent of the Judith Kleeman Professorial
  Chair.}
\and
Uri Stemmer\thanks{Ben-Gurion University and Google Research. \texttt{u@uri.co.il}. Partially supported by the Israel Science Foundation (grant No.\ 1871/19).}
}

\date{\today}
\maketitle
\setcounter{page}{0} \thispagestyle{empty}

\begin{abstract}
We study the sample complexity of learning threshold functions under the constraint of {\em differential privacy}. It is assumed that each labeled example in the training data is the information of one individual and we would like to come up with a generalizing hypothesis $h$ while guaranteeing differential privacy for the individuals. Intuitively, this means that any single labeled example in the training data should not have a significant effect on the choice of the hypothesis. This problem has received much attention recently; unlike the non-private case, where the sample complexity is {\em independent of the domain size} and just depends on the desired accuracy and confidence, for private learning the sample complexity must depend on the domain size $X$ (even for approximate differential privacy). Alon et al.\ (STOC 2019) showed a lower bound of $\Omega(\log^*|X|)$  on the sample complexity and
Bun et al.\ (FOCS 2015)  presented an approximate-private learner with sample complexity $\tilde{O}\left(2^{\log^*|X|}\right)$.
In this work we reduce this gap significantly,
almost settling the sample complexity. We first present a new upper bound (algorithm) of $\tilde{O}\left(\left(\log^*|X|\right)^2\right)$ on the sample complexity  and then present an improved version with sample complexity  $\tilde{O}\left(\left(\log^*|X|\right)^{1.5}\right)$.

Our algorithm is constructed for the related {\em interior point problem}, where the goal is to find a point between the largest and smallest input elements. It is based on selecting an input-dependent hash function and using it to embed the database into a  domain whose  size is reduced  logarithmically; this results in a new database, an interior point of which can be used to generate an interior point in the original database in a differentially private manner.
\end{abstract}

\end{titlepage}


\section{Introduction}

One of the most fundamental tasks in machine learning is that of learning 1-dimensional {\em threshold functions}. In this task,
we are given a collection of %
examples, called {\em training data},
where every example is taken from a finite domain $X\subseteq\R$ and is labeled by a fixed (but unknown) {\em threshold function}. (A threshold function is a binary function that evaluates to $1$ on some prefix of the domain.\footnote{Let $X\subseteq\R$. A threshold function $f$ over $X$ is specified by an element $u \in X$ so that $f(x)=1$ if $x \leq u$ and $f(x)=0$ for $x > u$.})
The goal is to generalize the training data into a hypothesis $h$ that predicts the labels of unseen examples.
In this paper we study this problem under the constraint of {\em differential privacy}. It is assumed that each labeled example in the training data is the information of one individual (e.g., every example might represent blood sugar level, and the label might indicate whether the individual has diabetes). We would like to come up with a generalizing hypothesis $h$ while guaranteeing differential privacy for the individuals. Intuitively, this means that any single labeled example in the training data should not have a significant effect on the choice of $h$
and in particular given the hypothesis it should be hard to distinguish whether an individual's data was used or not.

An important and natural measure for the efficiency of learning algorithms is the amount of data needed to produce a good hypothesis, a.k.a.\ the {\em sample complexity}. Without privacy constraints, learning thresholds is easy, and has a `constant' sample complexity that
depends only on the desired level of accuracy and confidence, but not on the domain size (i.e., one can simply set $X=\R$). With differential privacy, however, understanding the sample complexity of this supposedly simple problem has turned out to be quite challenging. Already in the first work on differentially private learning, Kasiviswanathan et al.~\cite{KLNRS08} presented a generic construction (obtained as a private variant of Occam's Razor~\cite{BlumerEHW87}) stating that the sample complexity of privately learning threshold functions over a domain $X$ is at most $O(\log|X|)$. That is, unlike the non-private sample complexity, %
the upper bound of~\cite{KLNRS08} grows logarithmically with the size of the domain.

This gap between the sample complexity of learning thresholds with or without privacy has received much attention since then.
For the case of {\em pure}-differential privacy (a strong variant of differential privacy), Feldman and Xiao~\cite{FX14} showed that this gap is unavoidable, and that every pure-private learner for thresholds over a domain $X$ must have sample complexity $\Omega(\log|X|)$. Beimel et al.~\cite{BNS13b} showed that the lower bound of~\cite{FX14} can be circumvented by relaxing the privacy requirement from pure to approximate-differential privacy. Specifically, they presented an approximate-private learner for threshold functions over a domain $X$ with sample complexity $\tilde{O}\left(8^{\log^*|X|}\right)$, a dramatic improvement in asymptotic terms over $\Theta(\log|X|)$. Bun et al.~\cite{BNSV15} then presented a different approximate-private learner with improved sample complexity of $\tilde{O}\left(2^{\log^*|X|}\right)$, and another different algorithm with similar sample complexity was presented by~\cite{BunDRS18}. 
Furthermore, Bun et al.~\cite{BNSV15} showed a lower bound of $\Omega(\log^*|X|)$ on the sample complexity of every approximate-private learner for thresholds that outputs a hypothesis that is itself a threshold function (such a learner is called {\em proper}). 
Recently, Alon et al.~\cite{AlonLMM19} showed that a lower bound of $\Omega(\log^*|X|)$ holds even for {\em improper} learners, i.e., for learners whose output hypothesis is not restricted to being a threshold function. To summarize, our current understanding of the task of privately learning thresholds places its sample complexity somewhere between $\Omega(\log^*|X|)$ and $\tilde{O}\left(2^{\log^*|X|}\right)$, a gap which is exponential in $\log^*|X|$, where at least three different algorithms are known with sample complexity $2^{O(\log^*|X|)}$.

In this work we reduce this gap significantly,
almost settling the sample complexity,  and present a new upper bound of $\tilde{O}\left(\left(\log^*|X|\right)^2\right)$. We then present an improvement of the algorithm with sample complexity  $\tilde{O}\left(\left(\log^*|X|\right)^{1.5}\right)$.
While the main goal of the work is to resolve the asymptotic sample complexity of the %
problem, our algorithm is computationally efficient and may be useful more generally as a technique for shrinking the data domain.

\paragraph{The interior point problem:} In order to obtain their results, Bun et al.~\cite{BNSV15} presented reductions in both directions between (properly) learning thresholds and the \emph{interior point problem}. Given a database $S$ containing (unlabeled) elements from $X$, the interior point problem asks for an element of $X$ between the smallest and largest elements in $S$.
\begin{definition}
An algorithm $\AAA$ solves the {\em interior point problem} over a domain $X$ with sample complexity $n$ and failure probability $\beta$ if for every database $S\in X^n$,
$$\Pr[\min S \leq \AAA(S)\leq \max S]\geq1-\beta,$$
where the probability is taken over the coins of $\AAA$. %
We call a solution $x$ with $\min S\leq x\leq\max S$ an {\em interior point} of $S$. Note that $x$ need not be a member of the database $S$.
\end{definition}
To see the equivalence between the interior point problem and (properly) learning thresholds, let $S\subseteq X$ be an input for the interior point problem, and construct a labeled database $D$ containing all elements of $S$ where the smallest $|S|/2$ elements are labeled by 1 and the largest $|S|/2$ elements are labeled by 0. Now consider applying a learner for thresholds on the database $D$, that returns a hypothesis $h$ that is itself a threshold function. The point at which $h$ switches from 1 to 0 must be an interior point of the database $S$. For the other direction, let $D$ be a labeled database, and suppose that $D$ contains both ``a lot'' of elements with the label 0 and ``a lot'' of elements with the label 1 (otherwise either $h\equiv0$ or $h\equiv1$ is a good output). Now construct an unlabeled database $S$ containing (say) the largest $|D|/10$ elements in $D$ which are labeled as 1 and the smallest $|D|/10$ elements in $D$ which are labeled as 0. Now consider applying an algorithm for the interior point problem on $S$ to obtain an outcome $y$, and define the hypothesis $f(x)=\1\{x\leq y\}$. It can be shown that this hypothesis has small error on $D$, completing the equivalence. %
In fact, Bun at al.~\cite{BNSV15} showed that the following four problems are equivalent (up to constant factors) under differential privacy:
\begin{enumerate}
\item The interior point problem.
\item Learning of threshold functions (properly).
\item Distribution learning (with respect to Kolmogorov distance).\footnote{Distribution learning is a fundamental problem in statistics. Given $n$ i.i.d.\ samples from an unknown distribution $\DDD$, the goal is to produce a distribution $\DDD'$ with small distance to $\DDD$.}
\item Query release for threshold functions.\footnote{Given a set $Q$ of queries $q:X^n\rightarrow\R$ the query release problem for $Q$ is to output accurate answers to all queries in Q. In the case of threshold functions, each query is specified by a domain element $x\in X$, and asks for the number of input element that are smaller than $x$.}
\end{enumerate}
Hence, showing new upper or lower bounds on the sample complexity of privately solving the interior point problem immediately results in new upper or lower bounds on the sample complexity of privately learning thresholds (properly), private distribution learning (w.r.t.\ Kolmogorov distance), and private query release for threshold functions.

\subsection{Overview of Our (First) Construction}
We design a new algorithm for privately solving the interior point problem with sample complexity $\tilde{O}\left(\left(\log^*|X|\right)^2\right)$. At a high level, the algorithm works by embedding the input elements from the domain $X$ in a smaller domain of size $\log|X|$, in such a way that guarantees that every interior point of the embedded elements can be (privately) translated into an interior point of the input elements. The algorithm is then applied recursively to identify an interior point of the embedded elements.

We now give an informal overview of the construction.
Let $X$ be a totally ordered domain, %
and consider a binary tree $T$ with $|X|$ leaves, where every leaf is identified with an element of $X$. We abuse notation and use $u_\ell$ (for $1\leq\ell\leq|X|$) to denote both the $\ell$th smallest element in $X$ and to denote the $\ell$th left-most leaf of the tree $T$.
Given an input database $S\in X^n$ containing $n$ elements from $X$, we assign {\em weights} to the nodes of $T$ in a bottom up manner. First, the weight $w(u_{\ell})$ of every leaf $u_{\ell}$ is its multiplicity in $S$. That is, $w(u_{\ell})=|\{x\in S: x=u_{\ell}\}|$. Now define the weight of every node $v$ in $T$ as the sum of the weights of its two children. For simplicity, we will assume throughout the introduction that every leaf has weight either 0 or 1 (i.e., we assume that no element appears twice in $S$; this is not true as we recurse).

We now explain how to generate the input database to the next recursive call. That is, how we embed the input elements in a domain of size $\log|X|$. To that end, suppose that we have generated  a {\em path} $\pi$ from the root to a leaf $u_{\pi}$ with positive weight (by our simplifying assumption this leaf has weight $1$).
Had we been able to compute such a path {\em in a private manner} then our task would be completed, since if $w(u_{\pi})>0$ then $u_{\pi}\in S$ and hence $u_{\pi}$ is an interior point of $S$, and there was no need to apply the recursion. Therefore we do not argue that we can release the path while maintaining differential privacy. Nevertheless, assume (for now) that such a path were given to us ``for free'', and instead of ending the computation by returning the leaf at the end of this path, we use it in order to construct the input database to the next recursive call.

Let $v$ be a node in the path $\pi$, let $v_{\scriptscriptstyle\rm next}$ be the next node in $\pi$, and let $v_{\scriptscriptstyle\rm other}$ denote the other child of $v$ in the tree $T$. We say that $m$ input points {\em fall off} the path $\pi$ at the node $v$ if $w(v_{\scriptscriptstyle\rm other})=m$. Input points that {\em fall off} the path $\pi$ at the node $v$ are input points that belong to the subtree rooted at $v$, but not to the subtree rooted at the next node in $\pi$.
As $\pi$ starts at a node with weight $n$ (the root of the tree) and ends at a leaf with weight 1 (the leaf $u_{\pi}$), we get that overall $n-1$ input points fall off the path $\pi$.

One can think of the path as acting as a hash function of the domain into a range that is logarithmic in its size, where a point is mapped to the level where it falls off the path. The key property we need is that adjacent databases will be embedded to adjacent databases. Initialize an empty database $D$ (which will be the input database to the next recursive call).
We add elements to $D$ by following the path $\pi$ from the root down. If $m$ points fall off the path at level $\ell$ of the tree, then we add $m$ copies of $\ell$ to the database $D$. Observe that $D$ contains elements from $[\log|X|]$, i.e., the domain size is reduced logarithmically.

Now suppose that (by recursion) we obtained an interior point $\ell^*$ of the database $D$, and let $v_{\pi}^{\ell^*}$ denote the node at level $\ell^*$ in the path $\pi$. Furthermore, let us assume for simplicity that there are in fact ``many'' points in $D$ that are bigger than $\ell^*$ and ``many'' points which are smaller than $\ell^*$, and so $\ell^*$ is a ``deep'' interior point of $D$.\footnote{Here ``many'' should be thought of as $\approx\frac{1}{\eps}\log\frac{1}{\delta}$; in the actual algorithm we ensure that these conditions indeed hold by `trimming' these elements from the database before the recursive call.} This means that ``many'' elements from $S$ fall off the path $\pi$ before level $\ell^*$ and that ``many'' elements fall off after level $\ell^*$ of the tree. Note that since ``many'' elements from $S$ fall off $\pi$ after level $\ell^*$, we get that the weight of $v_{\pi}^{\ell^*}$ (the node at level $\ell^*$ of the path $\pi$) is ``large''. %

Now let $v^*$ be a node at level $\ell^*$ of the tree that is chosen according to its weight in a differentially private manner (such a ``heavy'' node exists since $v_{\pi}^{\ell^*}$  is ``heavy'', and can be privately identified using standard {\em stability based} techniques). Now, since ``many'' elements from $S$ fall of the path $\pi$ before level $\ell^*$, it cannot be the case that all of the points in $S$ belong to the subtree rooted at $v^*$. Indeed, if that were the case then the path $\pi$ must go through $v^*$, i.e., $v^*=v_{\pi}^{\ell^*}$, but then all the elements of $S$ that fall off before level $\ell^*$ cannot belong to the subtree rooted at $v^*$. (For simplicity we assumed here that these elements fall {\em strictly} before level $\ell^*$; in the actual algorithm we handle such a case differently.)

The final output is then chosen from  one of two descendants of  $v^*$:  the left-most and right-most leaves of the sub-tree rooted at $v^*$, denoted as $v_{\scriptscriptstyle\rm left}$ and $v_{\scriptscriptstyle\rm right}$, respectively.\footnote{In the actual algorithm we also consider two additional descendants of  $v^*$.}
To see that at least one of these leafs is a good output, observe that since  $v^*$ has positive weight, then some of the elements of $S$ must belong to the subtree rooted at $v^*$. That is, the database $S$ contains elements between $v_{\scriptscriptstyle\rm left}$ and $v_{\scriptscriptstyle\rm right}$. In addition, since the subtree rooted at $v^*$ does not contain {\em all} elements of $S$, then either $S$ contains elements which are bigger than $v_{\scriptscriptstyle\rm right}$, in which case $v_{\scriptscriptstyle\rm right}$ is an interior point of $S$, or $S$ contains elements which are smaller than $v_{\scriptscriptstyle\rm left}$, in which case $v_{\scriptscriptstyle\rm left}$ is an interior point of $S$.

\paragraph{Selecting the path.} In the above description we assumed that the path $\pi$ is given to us ``for free''. While the path $\pi$ contains sensitive information and cannot be released privately, we show that it is possible to select a path $\pi$ randomly %
in such a way that guarantees the following condition. Let $S$ and $S'$ be two neighboring databases, and let $\DDD$ denote the distribution on databases that results from sampling a path $\pi$ in the execution on $S$ and embedding the points from $S$ using $\pi$. Similarly let $\DDD'$ denote this distribution w.r.t.\ $S'$. Then we show that these two distributions are ``close'' in the sense that for any database $D$ in the support of $\DDD$ (with non-negligible probability mass) there exists a {\em neighboring} database $D'$ with roughly the same probability mass under $\DDD'$. So, intuitively, while the paths $\pi$ that appear in the executions on $S$ and on $S'$ might be very different, we show that the resulting distributions on the databases for the next recursive call are ``similar''. This can be formalized to show that the algorithm satisfies differential privacy.

We obtain the following theorem, which we present in Section~\ref{sec:TreeLog}.

\begin{theorem}\label{thm:introPowerTwo}
Let $X$ be a totally ordered domain. There exists an $(\eps,\delta)$-differentially private algorithm that solves the interior point problem on $X$ with success probability 9/10 and sample complexity $n=\tilde{O}\left(\frac{1}{\eps}\cdot\log(\frac{1}{\delta})\cdot(\log^*|X|)^2\right).$
\end{theorem}

In Section~\ref{sec:decompose} we present an improved analysis of (a variant of) this algorithm, which results in a sample complexity of $n=\tilde{O}\left(\frac{1}{\eps}\cdot\log^{1.5}(\frac{1}{\delta})\cdot(\log^*|X|)^{1.5}\right)$.
Using the equivalence of Bun et al.~\cite{BNSV15}, our new algorithm for the interior point problem results in new algorithms for privately learning threshold functions, for private distribution learning (w.r.t.\ Kolmogorov distance), and for private query release of threshold functions, all with sample complexity proportional to $\tilde{O}\left(\left(\log^*|X|\right)^{1.5}\right)$. Previous algorithms for these tasks had sample complexity proportional to $\tilde{O}\left(2^{\log^*|X|}\right)$.


\section{Preliminaries}

An algorithm operating on databases is said to preserve differential privacy if changing a single record of its input database does not significantly change the output
distribution of the algorithm. Intuitively, this means that whatever is learned about an individual could also
be learned with her data arbitrarily modified (or without her data). Formally:

\begin{definition}[\cite{DMNS06,DKMMN06}]\label{def:dp}
Two databases are called {\em neighboring} if they differ in a single entry. A randomized algorithm $\AAA$ is $(\eps,\delta)$-{\em differentially private} if for every two neighboring databases $S,S'$ and for any event $T$,
$$\Pr[\AAA(S)\in T]\leq e^{\eps}\cdot \Pr[\AAA(S')\in T]+\delta.$$ 
If $\delta==0$ the type of privacy guarantee is  {\em pure} and if $0<\delta<1$ it is {\em approximate}. 
\end{definition}

One of the key reasons differential privacy is such a powerful notion is the composition properties it enjoys. In particular, in ``advanced composition" the adaptive application of $k$ mechanisms, each of which is $(\eps,\delta)$-differentially private, satisfies $\approx(\eps\sqrt{k},\delta k)$-differential privacy~\cite{DRV10}.
For background on differential privacy see Dwork and Roth~\cite{DR14} or Vadhan~\cite{Vadhan2016}.
\subsection{The Exponential and Choosing Mechanisms}
Let $X^*$ denote the set of all finite databases over a domain $X$. A quality function  $q:X^*\times \FFF \rightarrow \N$ defines an {\em optimization problem} over the domain $X$ and a finite solution set $\FFF$:  Given a database $S \in X^*$, choose $f\in\FFF$ that (approximately) maximizes $q(S,f)$. We say that the function $q$ has {\em sensitivity} $\Delta$ if for all neighboring databases $S$ and $S'$ and for all $f\in\FFF$ we have $|q(S,f)-q(S',f)|\leq \Delta$.

The Exponential Mechanism of McSherry and Talwar~\cite{McSherryTa07} solves such optimization problems by choosing a random solution where the probability of outputting any solution $f$ increases exponentially with its quality $q(S,f)$. Specifically, it outputs each $f \in \FFF$ with probability proportional to $\exp\left(\eps \cdot q(S,f) /(2 \Delta)\right)$. The privacy and utility of the mechanism are expressed as:

\begin{lemma}[\cite{McSherryTa07}] \label{prop:exp_mech}
The Exponential Mechanism is $(\eps,0)$-differentially private.
Let $q$ be a quality function with sensitivity at most $1$. Fix a database $S \in X^n$ and let $\OPT = \max_{f\in \FFF}\{q(S,f)\}$.
With probability at least $(1-\beta)$, the Exponential Mechanism outputs a solution $f$ with quality
$q(S,f)\geq\OPT-\frac{2}{\eps}\ln\left(\frac{|\FFF|}{\beta}\right)$.
\end{lemma}

Note that the quality of the solution guaranteed by the Exponential Mechanism decreases with $\log|\FFF|$. For a sub family of low-sensitivity functions, called {\em bounded-growth} functions, this can be avoided by using an $(\eps, \delta)$-differentially private variant of the Exponential Mechanism called the \emph{Choosing Mechanism}, introduced by Beimel et al.~\cite{BNS13b}. A quality function with sensitivity at most $1$ is of {\em $k$-bounded-growth} if adding an element to a database can increase (by 1) the score of at most $k$ solutions, without changing the scores of other solutions. Specifically, it holds that
\begin{enumerate}
\item $q(\emptyset, f) = 0$ for all $f \in \FFF$,
\item If $S_2 = S_1 \cup \{x\}$, then $q(S_1, f) + 1 \ge q(S_2, f) \ge q(S_1, f)$ for all $f \in \FFF$, and
\item There are at most $k$ values of $f$ for which $q(S_2, f) = q(S_1, f) + 1$.
\end{enumerate}

\begin{example}
Consider a case where $\FFF=X$ and $q(S,f)=|\{x\in S : x=f\}|$. That is, the quality of a solution (=domain element) $f$ is its multiplicity in the database $S$. Note that with this quality function we are aiming to identify an element $f\in X$ with large multiplicity in $S$. Observe that this function is 1-bounded growth.
\end{example}

The Choosing Mechanism is a differentially private algorithm for approximately solving bounded-growth choice problems. The following lemma specifies its privacy and utility guarantees.

\begin{lemma}[\cite{BNS13b,BNSV15}]\label{lem:CM}
Let $\delta > 0$, and $0 < \eps\leq2$.
The Choosing Mechanism is $(\eps,\delta)$-differentially private.
Let the Choosing Mechanism be executed on a $k$-bounded-growth quality function, and on a database $S$ containing $n$ elements. Denote $\OPT = \max_{f\in \FFF}\{q(S,f)\}$. With probability at least $(1-\beta)$, the Choosing Mechanism outputs a solution $f$ with quality
$q(S,f)\geq\OPT-\frac{16}{\eps}\ln(\frac{4kn}{\beta\eps\delta})$.
\end{lemma}


\section{Algorithm \texttt{TreeLog} and its Analysis}\label{sec:TreeLog}
In this section we present algorithm \texttt{TreeLog}, which privately solves the interior point problem with sample complexity $\tilde{O}\left(\left(\log^*|X|\right)^2\right)$. Consider algorithm  \texttt{TreeLog}, described in Algorithm~\ref{alg:TreeLog}. We analyze its privacy properties in Section~\ref{sec:privacy_analysis} and analyze its utility properties in Section~\ref{sec:utility_analysis}. For an illustration of the algorithm see Figure~\ref{fig:tree}.

\begin{algorithm*}[!htp]

\caption{\bf \texttt{TreeLog}}\label{alg:TreeLog}

{\bf Input:} Parameters $\eps,\delta$ and a database $S\in X^n$ where $X$ is a totally ordered domain. We assume that $|X|$ is a power of 2 (otherwise extend $X$ to the closest power of 2).
Set trimming parameter  $t = \Theta\left(\frac{1}{\eps}\log\frac{1}{\delta}\right)$.

\begin{enumerate}[leftmargin=15pt,rightmargin=10pt,itemsep=1pt,topsep=1.5pt]

\item\label{step:baseCase} If $|X|=O(1)$ 
then use the Exponential Mechanism to return $y\in X$ with large quality
$$q(S,y)=\min\left\{\;|\{x\in S: x\leq y\}|,\; |\{x\in S: x\geq y\}|\;\right\}.$$

\item\label{step:takeMiddle}
Sort $S$ and let $\hat{S}\in X^{n-2t}$ be a database containing all elements of $S$ except for the $t$ largest and $t$ smallest elements.

\item\label{step:buildTree} Let $T$ be a complete binary tree with $|X|$ leaves that correspond to elements of $X$. A leaf $u$ has weight $w(u)=|\{x\in \hat{S}: x=u\}|$. A node $v$ has weight $w(v)$ that equals the sum of the weights of its children.

\item\label{step:randomPath} Sample a path $\pi$ in $T$, starting from the root and constructed by the following process:
\begin{enumerate}[topsep=0pt]
	\item Let $v$ be the current node in the path.
	\item If $v$ is a leaf, or if $w(v)\leq t$, then $v$ is the last node in $\pi$.
	\item Else, let $v_0,v_1$ be the two children of $v$ in $T$. If one of them has weight 0 then proceed to the other child. Otherwise, proceed to $v_b$ with probability proportional to $\exp(\eps\cdot w(v_b))$.
\end{enumerate}

\item\label{step:buildD} Initialize $D=\emptyset$. Add elements to $D$ by following the path $\pi$ starting from the root:
\begin{enumerate}[topsep=0pt]
	\item\label{step:buildDa} If $|D|\geq n-3t$ then goto Step~\ref{step:recursiveCall}.
	\item Let $v$ be the current node in the path, and let $\ell$ denote its level in $T$ (the root is in level 0 and the leaves are in level $\log|X|$).
	\item If $v$ is the last node in $\pi$ then add $(n-3t-|D|)$ copies of $\ell$ to $D$ and goto Step~\ref{step:recursiveCall}.
	\item Else, let $v_{\scriptscriptstyle\rm next}$ be the next node in $\pi$, and let $v_{\scriptscriptstyle\rm other}$ be the other child of $v$ in $T$. Add $\min\left\{w(v_{\scriptscriptstyle\rm other}),n-3t-|D|\right\}$ copies of $\ell$ to $D$, and goto Step~\ref{step:buildDa} with $v_{\scriptscriptstyle\rm next}$ as the current node.
\end{enumerate}

\item\label{step:recursiveCall} Execute \texttt{TreeLog} recursively on $D$ and let $\ell^*$ denote the returned outcome.

\item\label{step:choosingMech} Use the Choosing Mechanism with privacy parameters $\eps,\delta$ to choose a node $v^*$ at level $\ell^*$ of $T$ with large weight $w(v^*)$.

\item\label{step:identifyLeaves} Let $v_{\scriptscriptstyle\rm left}$ and $v_{\scriptscriptstyle\rm right}$ be the left-most and right-most leaves, respectively, of the sub-tree rooted at $v^*$. Also let $v_{\scriptscriptstyle\rm inner\text{-}left}$ be the right-most leaf of the sub-tree rooted at the left child of $v^*$, and let $v_{\scriptscriptstyle\rm inner\text{-}right}$ be the left-most leaf of the sub-tree rooted at the right child of $v^*$.

\item\label{step:expMech} Use the Exponential Mechanism with privacy parameter $\eps$ to return $y\in\left\{v_{\scriptscriptstyle\rm left}, v_{\scriptscriptstyle\rm right}, v_{\scriptscriptstyle\rm inner\text{-}left}, v_{\scriptscriptstyle\rm inner\text{-}right} \right\}$ with large quality
$$q(S,y)=\min\left\{\;|\{x\in S: x\leq y\}|,\; |\{x\in S: x\geq y\}|\;\right\}.$$

\end{enumerate}
\end{algorithm*}

\tikzset{every picture/.style={line width=0.75pt}} 

\begin{figure}[!b]
\label{fig:tree}

\begin{tikzpicture}[x=0.75pt,y=0.75pt,yscale=-1,xscale=1]

\draw  [color={rgb, 255:red, 65; green, 117; blue, 5 }  ,draw opacity=1 ][line width=2.25]  (299,25.5) .. controls (299,15.28) and (307.28,7) .. (317.5,7) .. controls (327.72,7) and (336,15.28) .. (336,25.5) .. controls (336,35.72) and (327.72,44) .. (317.5,44) .. controls (307.28,44) and (299,35.72) .. (299,25.5) -- cycle ;
\draw  [color={rgb, 255:red, 65; green, 117; blue, 5 }  ,draw opacity=1 ][line width=2.25]  (157,89.5) .. controls (157,79.28) and (165.28,71) .. (175.5,71) .. controls (185.72,71) and (194,79.28) .. (194,89.5) .. controls (194,99.72) and (185.72,108) .. (175.5,108) .. controls (165.28,108) and (157,99.72) .. (157,89.5) -- cycle ;
\draw  [color={rgb, 255:red, 65; green, 117; blue, 5 }  ,draw opacity=1 ][line width=2.25]  (63,150.5) .. controls (63,140.28) and (71.28,132) .. (81.5,132) .. controls (91.72,132) and (100,140.28) .. (100,150.5) .. controls (100,160.72) and (91.72,169) .. (81.5,169) .. controls (71.28,169) and (63,160.72) .. (63,150.5) -- cycle ;
\draw  [color={rgb, 255:red, 65; green, 117; blue, 5 }  ,draw opacity=1 ][line width=2.25]  (141,208.5) .. controls (141,198.28) and (149.28,190) .. (159.5,190) .. controls (169.72,190) and (178,198.28) .. (178,208.5) .. controls (178,218.72) and (169.72,227) .. (159.5,227) .. controls (149.28,227) and (141,218.72) .. (141,208.5) -- cycle ;
\draw  [color={rgb, 255:red, 65; green, 117; blue, 5 }  ,draw opacity=1 ][line width=2.25]  (98,262.5) .. controls (98,252.28) and (106.28,244) .. (116.5,244) .. controls (126.72,244) and (135,252.28) .. (135,262.5) .. controls (135,272.72) and (126.72,281) .. (116.5,281) .. controls (106.28,281) and (98,272.72) .. (98,262.5) -- cycle ;
\draw   (458,89.5) .. controls (458,79.28) and (466.28,71) .. (476.5,71) .. controls (486.72,71) and (495,79.28) .. (495,89.5) .. controls (495,99.72) and (486.72,108) .. (476.5,108) .. controls (466.28,108) and (458,99.72) .. (458,89.5) -- cycle ;
\draw   (397,150.5) .. controls (397,140.28) and (405.28,132) .. (415.5,132) .. controls (425.72,132) and (434,140.28) .. (434,150.5) .. controls (434,160.72) and (425.72,169) .. (415.5,169) .. controls (405.28,169) and (397,160.72) .. (397,150.5) -- cycle ;
\draw   (357,208.5) .. controls (357,198.28) and (365.28,190) .. (375.5,190) .. controls (385.72,190) and (394,198.28) .. (394,208.5) .. controls (394,218.72) and (385.72,227) .. (375.5,227) .. controls (365.28,227) and (357,218.72) .. (357,208.5) -- cycle ;
\draw   (516,148.5) .. controls (516,138.28) and (524.28,130) .. (534.5,130) .. controls (544.72,130) and (553,138.28) .. (553,148.5) .. controls (553,158.72) and (544.72,167) .. (534.5,167) .. controls (524.28,167) and (516,158.72) .. (516,148.5) -- cycle ;
\draw   (321,263.5) .. controls (321,253.28) and (329.28,245) .. (339.5,245) .. controls (349.72,245) and (358,253.28) .. (358,263.5) .. controls (358,273.72) and (349.72,282) .. (339.5,282) .. controls (329.28,282) and (321,273.72) .. (321,263.5) -- cycle ;
\draw   (494.5,263.75) .. controls (494.5,252.84) and (503.34,244) .. (514.25,244) .. controls (525.16,244) and (534,252.84) .. (534,263.75) .. controls (534,274.66) and (525.16,283.5) .. (514.25,283.5) .. controls (503.34,283.5) and (494.5,274.66) .. (494.5,263.75) -- cycle ;
\draw   (430,265.5) .. controls (430,255.28) and (438.28,247) .. (448.5,247) .. controls (458.72,247) and (467,255.28) .. (467,265.5) .. controls (467,275.72) and (458.72,284) .. (448.5,284) .. controls (438.28,284) and (430,275.72) .. (430,265.5) -- cycle ;
\draw   (604,262.5) .. controls (604,252.28) and (612.28,244) .. (622.5,244) .. controls (632.72,244) and (641,252.28) .. (641,262.5) .. controls (641,272.72) and (632.72,281) .. (622.5,281) .. controls (612.28,281) and (604,272.72) .. (604,262.5) -- cycle ;
\draw [color={rgb, 255:red, 65; green, 117; blue, 5 }  ,draw opacity=1 ][line width=2.25]    (95,136.5) -- (158,100.5) ;

\draw [color={rgb, 255:red, 65; green, 117; blue, 5 }  ,draw opacity=1 ][line width=2.25]    (193,79.5) -- (298,34.5) ;

\draw [color={rgb, 255:red, 65; green, 117; blue, 5 }  ,draw opacity=1 ][line width=2.25]    (93,166.5) -- (143,199.5) ;

\draw [color={rgb, 255:red, 65; green, 117; blue, 5 }  ,draw opacity=1 ][line width=2.25]    (129,250.5) -- (152,225.5) ;

\draw [color={rgb, 255:red, 0; green, 0; blue, 0 }  ,draw opacity=1 ][line width=2.25]    (335,34.5) -- (459,80.5) ;

\draw [color={rgb, 255:red, 0; green, 0; blue, 0 }  ,draw opacity=1 ][line width=2.25]    (431,136.5) -- (463,106) ;

\draw [color={rgb, 255:red, 0; green, 0; blue, 0 }  ,draw opacity=1 ][line width=2.25]    (382,191.5) -- (404,165.5) ;

\draw [color={rgb, 255:red, 0; green, 0; blue, 0 }  ,draw opacity=1 ][line width=2.25]    (349,246.5) -- (366,225.5) ;

\draw [color={rgb, 255:red, 0; green, 0; blue, 0 }  ,draw opacity=1 ][line width=2.25]    (524,133) -- (489,103.5) ;

\draw (67,159.5) node  [align=left] {};
\draw (37,84.5) node  [align=left] {Level $\displaystyle \ell^{*}$};
\draw (529,85.5) node  [align=left] {Node $\displaystyle v^{*}$};
\draw (285,159.5) node  [align=left] {$ $};
\draw (312,234.5) node  [align=left] { $\displaystyle v_{\rm left}$};
\draw (440,232.5) node  [align=left] { $\displaystyle v_{\rm inner\text{-}left}$};
\draw (220,36) node  [align=left] {Path $\displaystyle \pi $};
\draw (633,231.5) node  [align=left] { $\displaystyle v_{\rm right}$};
\draw (547,231.5) node  [align=left] { $\displaystyle v_{\rm inner\text{-}right}$};

\end{tikzpicture}
\caption{An illustration of the objects used by Algorithm \texttt{TreeLog}.}
\label{fig:tree}
\end{figure}
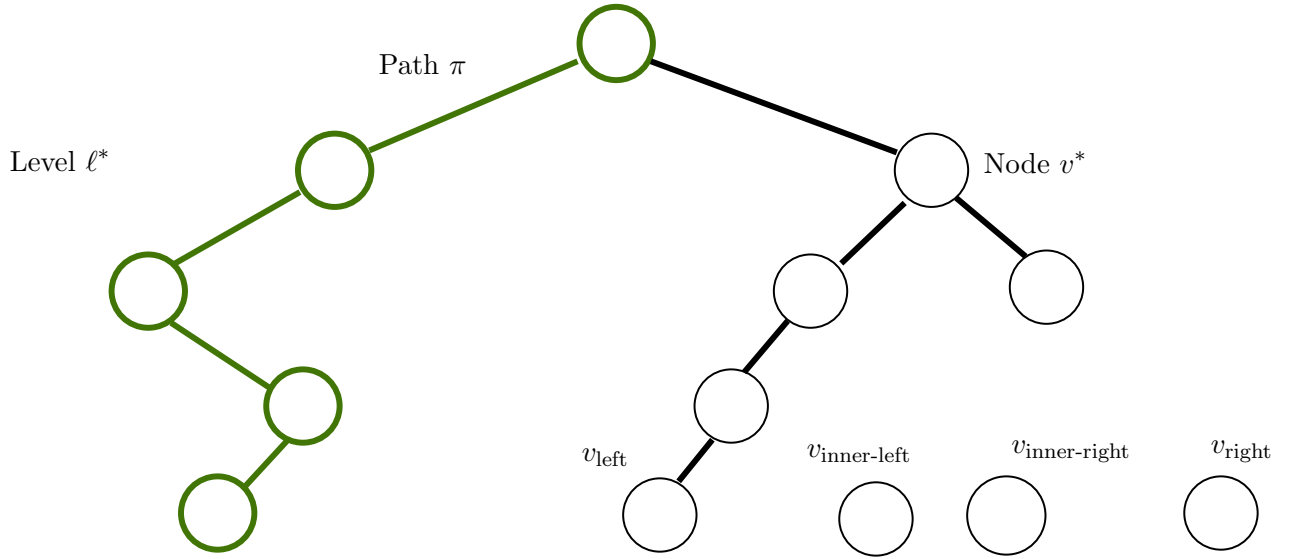

\subsection{Privacy Analysis}
\label{sec:privacy_analysis}

First, observe that it suffices to analyze the algorithm without Step~\ref{step:takeMiddle}. The reason is that for any two neighboring databases $S,S'$ we have that the multisets $\hat{S},\hat{S'}$ containing the middle elements of $S,S'$, respectively, are neighboring multisets (this technique was also used in~\cite{BNSV15}). So, for the privacy analysis, we will assume that $\hat{S}$ contains all elements of $S$.

Fix two neighboring databases $S$ and $S'=S\cup\{x'\}$, and consider the (top-level call of the) execution of \texttt{TreeLog} on $S$ and on $S'$. We use $w_S(v)$ and $w_{S'}(v)$ to denote the weight of a node $v$ during the two executions, respectively. Also, let $P(S)$ and $P(S')$ denote the set of possible paths that can be obtained in Step~\ref{step:randomPath} of the execution on $S$ and on $S'$, respectively. That is, $P(S)$ denotes the support of the distribution on paths defined in Step~\ref{step:randomPath} of the (top-level call of the) execution on $S$.
For a path $\pi$ we write $\Pr_S[\pi]$ and $\Pr_{S'}[\pi]$ to denote the probability of obtaining this path in the executions on $S$ and on $S'$, respectively. We make the following observations.

\begin{observation}\label{obs:atmostNpaths}
$|P(S)|\leq |S|$ and $|P(S')|\leq|S'|$.
\end{observation}
This is because paths only end in nodes with positive weight, and because if $\pi_1\neq\pi_2$ are both in $P(S)$ then $\pi_1$ is not a prefix of $\pi_2$ and vice verse.

\begin{observation}
Let $\pi\in P(S)$, and let $v_{\pi}$ denote the last node in $\pi$. There exists either one or two paths in $P(S')$ such that $\pi$ is their prefix.
\end{observation}
This is because if $v_{\pi}$ is reachable from the root during the execution on $S$ then it is also reachable during the execution on $S'$. When reaching $v_{\pi}$ in the execution on $S'$, the path ends in $v_{\pi}$ if $w_{S'}(v_{\pi})\leq t$, and otherwise the path continues and can split at most once into two different paths.

\begin{definition}
A path $\pi\in P(S)$ is {\em weak} for $S$ if $\Pr_S[\pi]\leq\delta$. Otherwise, the path is {\em strong} for $S$. We use similar definitions w.r.t.\ $S'$.
\end{definition}
\begin{observation}
Let $\pi'\in P(S')$. Either there is a path $\pi\in P(S)$ such that $\pi$ is a prefix of $\pi'$, or else $\pi'$ is weak for $S'$.
\end{observation}

\begin{proof}
Assume that no path in $P(S)$ is a prefix of $\pi'$.
Let $v_{\scriptscriptstyle\rm last}$ denote the last node in the path $\pi'$ that appears in a path of $P(S)$, and let $v_{\scriptscriptstyle\rm break}$ denote the next node in $\pi'$. Note that $v_{\scriptscriptstyle\rm break}$ must exist because otherwise $\pi'\in P(S)$.  
Since no path in $P(S)$ is a prefix of $\pi'$, we get that any path in $P(S)$ that reaches $v_{\scriptscriptstyle\rm last}$ does not stop in $v_{\scriptscriptstyle\rm last}$, and hence, $w_S(v_{\scriptscriptstyle\rm last})>t$ and therefore $w_{S'}(v_{\scriptscriptstyle\rm last})>t$. 
Moreover, since $w_S(v_{\scriptscriptstyle\rm last})>t$ and since $v_{\scriptscriptstyle\rm break}$ is not on a path of $P(S)$, we have that $w_S(v_{\scriptscriptstyle\rm break})=0$ and hence $w_{S'}(v_{\scriptscriptstyle\rm break})=1$. 
So $w_{S'}(v_{\scriptscriptstyle\rm break})=1$ and $w_{S'}(v_{\scriptscriptstyle\rm last})>t$ and therefore the probability of proceeding from $v_{\scriptscriptstyle\rm last}$ to $v_{\scriptscriptstyle\rm break}$ is at most $\delta$.
\end{proof}

\begin{claim}\label{claim:pathsProbabilities}
Let $\pi\in P(S)$ be a {\em strong} path, and let $\Pi'\subseteq P(S')$ be a subset containing all paths in $P(S')$ such that $\pi$ is their prefix. Then $\Pr_S[\pi]\leq e^{3\eps\cdot\log n}\cdot\Pr_{S'}[\Pi']$.
\end{claim}

\begin{proof}
Let $v_0$ denote the first node in $\pi$ (the root of the tree $T$) and let $v_{\pi}$ denote the last node in $\pi$.
Also let $v_1,v_2,\dots,v_k$ denote all the nodes in the path $\pi$ (different from $v_{\pi}$) such that both of their children in $T$ have positive weight in the execution on $S$ (these are the nodes from which the construction of the path in Step~\ref{step:randomPath} is done randomly). For every such node $v_i$ let $v_{i,{\scriptscriptstyle\rm next}}$ denote its following node in $\pi$, and let $v_{i,{\scriptscriptstyle\rm other}}$ denote $v_i$'s other child in $T$. Note that $v_{k,{\scriptscriptstyle\rm next}}=v_{\pi}$.
So $\pi$ can be written as
$$\pi:=v_0%
\leadsto v_1\rightarrow v_{1,{\scriptscriptstyle\rm next}} \leadsto v_2\rightarrow v_{2,{\scriptscriptstyle\rm next}} \leadsto \dots \leadsto v_k \rightarrow v_{\pi}.$$
With these notations we have that
$$
\Pr\nolimits_S[\pi]=\frac{e^{\eps\cdot w_S(v_{1,{\scriptscriptstyle\rm next}})}}{e^{\eps\cdot w_S(v_{1,{\scriptscriptstyle\rm next}})} + e^{\eps\cdot w_S(v_{1,{\scriptscriptstyle\rm other}})}}\cdot \frac{e^{\eps\cdot w_S(v_{2,{\scriptscriptstyle\rm next}})}}{e^{\eps\cdot w_S(v_{2,{\scriptscriptstyle\rm next}})} + e^{\eps\cdot w_S(v_{2,{\scriptscriptstyle\rm other}})}} \dots \frac{e^{\eps\cdot w_S(v_{k,{\scriptscriptstyle\rm next}})}}{e^{\eps\cdot w_S(v_{k,{\scriptscriptstyle\rm next}})} + e^{\eps\cdot w_S(v_{k,{\scriptscriptstyle\rm other}})}}.
$$
During the execution on $S'$, some of these weights could be larger by (at most) 1. In addition, there might be another (at most one) node along the path $\pi$ such that the weight of both its children is positive, but in that case the weight of one of its children is exactly 1 (since during the execution on $S$ the weight of this node is 0). Hence, the probability of proceeding from this node along $\pi$ (during the execution on $S'$) is at least $(1-\delta)$. Therefore,
\begin{equation}
\Pr\nolimits_{S'}[\Pi]\geq(1-\delta)\cdot\frac{e^{\eps\cdot w_S(v_{1,{\scriptscriptstyle\rm next}})}}{e^{\eps\cdot[w_S(v_{1,{\scriptscriptstyle\rm next}})+1]} + e^{\eps\cdot[w_S(v_{1,{\scriptscriptstyle\rm other}})+1]}} \dots \frac{e^{\eps\cdot w_S(v_{k,{\scriptscriptstyle\rm next}})}}{e^{\eps\cdot[w_S(v_{k,{\scriptscriptstyle\rm next}})+1]} + e^{\eps\cdot[w_S(v_{k,{\scriptscriptstyle\rm other}})+1]}}.
\label{eq:1}
\end{equation}
Because the path $\pi$ is {\em strong}, for every $i$ it holds that $w_S(v_{i,{\scriptscriptstyle\rm next}})>w_S(v_{i,{\scriptscriptstyle\rm other}})-\frac{1}{\eps}\log(\frac{1}{\delta})$, as otherwise the probability of proceeding from $v_i$ to $v_{i,{\scriptscriptstyle\rm next}}$ would be at most $\delta$ and the path would not be strong.
In addition, for every $i$ such that $w_S(v_{i,{\scriptscriptstyle\rm next}})>w_S(v_{i,{\scriptscriptstyle\rm other}})+\frac{1}{\eps}\log(\frac{1}{\delta})$ we have that the term corresponding to $i$ in Inequality~(\ref{eq:1}) is at least $(1-\delta)$.
Finally, note that there could be at most $2\log n$ indices $i$ such that $\left|w_S(v_{i,{\scriptscriptstyle\rm next}})-w_S(v_{i,{\scriptscriptstyle\rm other}})\right|\leq \frac{1}{\eps}\log(\frac{1}{\delta})$. This is because the path $\pi$ ends when reaching a node of weight at most $t=\Theta(\frac{1}{\eps}\log\frac{1}{\delta})$, and every step in which $\left|w_S(v_{i,{\scriptscriptstyle\rm next}})-w_S(v_{i,{\scriptscriptstyle\rm other}})\right|\leq \frac{1}{\eps}\log(\frac{1}{\delta})$ decreases the weight of the next node by at least $1/3$ if $t\geq\frac{2}{\eps}\log(\frac{1}{\delta})$. Putting these observations together, and using the fact that that $k\leq n$ (which follows from Observation~\ref{obs:atmostNpaths}), we get that
\begin{align*}
\Pr\nolimits_{S'}[\Pi]&\geq(1-\delta)\cdot(1-\delta)^n \cdot e^{-2\eps\log n}\cdot \frac{e^{\eps\cdot w_S(v_{1,{\scriptscriptstyle\rm next}})}}{e^{\eps\cdot w_S(v_{1,{\scriptscriptstyle\rm next}}) } + e^{\eps\cdot w_S(v_{1,{\scriptscriptstyle\rm other}})}} \dots \frac{e^{\eps\cdot w_S(v_{k,{\scriptscriptstyle\rm next}})}}{e^{\eps\cdot w_S(v_{k,{\scriptscriptstyle\rm next}}) } + e^{\eps\cdot w_S(v_{k,{\scriptscriptstyle\rm other}}) }}\\
&=(1-\delta)\cdot(1-\delta)^n \cdot e^{-2\eps\log n}\cdot \Pr\nolimits_{S}[\pi]\\
&\geq e^{-4\delta n}\cdot e^{-2\eps\log n}\cdot \Pr\nolimits_{S}[\pi]\\
&\geq e^{-3\eps\log n}\cdot \Pr\nolimits_{S}[\pi],
\end{align*}
where the last inequality holds when $\delta\leq\frac{\eps}{4n}$.
\end{proof}

\begin{lemma}\label{lem:TreeLogPrivacyInduction}
An execution of algorithm \texttt{TreeLog} with $N$ recursive calls is $(5\eps N\log n, 3\delta n N e^{3\eps N\log n})$-differentially private.
\end{lemma}

\begin{proof}
Assume towards induction that the lemma holds whenever the algorithm performs at most $N-1$ recursive calls, and let $X$ be a domain that causes the algorithm to perform $N$ recursive calls. Let $S\in X^n$ and let $S'=S\cup\{x'\}$. Fix a set $F\subseteq X$ of possible outcomes. We will show that $\Pr[\texttt{TreeLog}(S)\in F]\leq e^{5\eps N \log n}\cdot\Pr[\texttt{TreeLog}(S')\in F]+3\delta n N e^{3\eps N\log n}$. The other direction follows from similar arguments.

Let $\pi\in P(S)$, and let $\Pi'\subseteq P(S')$ be a subset containing all paths in $P(S')$ such that $\pi$ is their prefix. We first show that the following inequality holds:
\begin{equation}
\Pr[\texttt{TreeLog}(S)\in F | \pi] \leq e^{5\eps (N-1)\log(n)+2\eps}\cdot \Pr[\texttt{TreeLog}(S')\in F | \Pi']+ 3\delta n (N-1) e^{3\eps (N-1)\log n}+2\delta.
\label{eq:2}
\end{equation}

To see this, let $D$ and $D'$ denote the multisets constructed in Step~\ref{step:buildD} of the execution on $S$ and on $S'$ respectively. Now note that when fixing $\pi$ in the execution on $S$, and any $\pi'\in\Pi'$ in the execution on $S'$, the resulting multisets $D,D'$ are neighboring. Specifically, if $\pi=\pi'$ then this is clear
(since the number of elements that fall off $\pi$ in $S$ and $S'$ differ by at most one in a single position).
Now suppose that $\pi\neq \pi'$ (but $\pi$ is still a prefix of $\pi'$). Let $v_{\pi}$ denote the last node in $\pi$. Since $\pi$ ends in $v_{\pi}$ but $\pi'$ does not, we have that $w_S(v_{\pi})=t$ and $w_{S'}(v_{\pi})=t+1$.
This means that the part of $D$ generated before we reach $v_{\pi}$ when running on $S$ is identical to the part of $D'$ generated before we reach $v_{\pi}$ when running on $S'$,
because the weight of all the nodes that fall off $\pi$ and $\pi'$ before $v_{\pi}$ is the same in $S$ and in $S'$. Since $w_{S'}(v_{\pi})=t+1$ and since we stop adding elements to $D$ once the remaining weight is $t$, we get that exactly one more element will be added to $D'$. %

Therefore, once we fix $\pi$ and $\pi'$, Step~\ref{step:recursiveCall} satisfies $(5\eps (N-1)\log n, 3\delta n (N-1) e^{3\eps (N-1)\log n})$-differential privacy by the induction assumption. In Steps~\ref{step:choosingMech} and~\ref{step:expMech} we apply the Choosing Mechanism and the Exponential Mechanism, each of which satisfies $(\eps,\delta)$-differential privacy. Inequality~\ref{eq:2} follows from (simple) composition.

The proof now follows by the following inequality.
{\small
\begin{align*}
&\Pr[\texttt{TreeLog}(S)\in F] = \sum_{\pi\in P(S)}\Pr\nolimits_S[\pi]\cdot\Pr[\texttt{TreeLog}(S)\in F|\pi]\\
&\quad\leq \sum_{{\rm weak }\; \pi}\Pr\nolimits_S[\pi]\cdot\Pr[\texttt{TreeLog}(S)\in F|\pi]
+ \sum_{{\rm strong }\; \pi}\Pr\nolimits_S[\pi]\cdot\Pr[\texttt{TreeLog}(S)\in F|\pi]\\
&\quad\leq n\delta + \sum_{{\rm strong }\; \pi}\Pr\nolimits_S[\pi]\cdot\Pr[\texttt{TreeLog}(S)\in F|\pi]\\
&\quad\leq n\delta + \sum_{{\rm strong }\; \pi}e^{3\eps\log n}\cdot\Pr\nolimits_{S'}[\Pi']\cdot\left(e^{5\eps (N-1)\log(n)+2\eps}\cdot\Pr[\texttt{TreeLog}(S')\in F|\Pi']+3\delta n (N-1) e^{3\eps (N-1)\log n}+2\delta\right)\\
&\quad\leq n\delta +3\delta n (N-1) e^{3\eps N\log n} +2\delta e^{3\eps N\log n} + \sum_{{\rm strong }\; \pi}e^{5\eps N\log(n)}\cdot\Pr\nolimits_{S'}[\Pi']\cdot\Pr[\texttt{TreeLog}(S')\in F|\Pi']\\
&\quad\leq 3\delta n N e^{3\eps N\log n} + e^{5\eps N\log(n)}\cdot\Pr[\texttt{TreeLog}(S')\in F],
\end{align*}
}
where the third inequality follows from Claim~\ref{claim:pathsProbabilities} and from Inequality~\ref{eq:2}.
\end{proof}

Recall that each recursive call shrinks the domain size logarithmically, and hence, on a database $S\in X^*$ algorithm \texttt{TreeLog} preforms at most $\log^*|X|$ recursive calls.
The following lemma, therefore, follows directly from Lemma~\ref{lem:TreeLogPrivacyInduction}.

\begin{lemma}\label{lem:TreeLogPrivacy}
Let algorithm \texttt{TreeLog} be executed on databases containing $n$ elements from a domain $X$. The algorithm is $(5\eps \log^*|X| \log n, 3\delta n \log^*|X| e^{3\eps \log^*|X| \log n})$-differentially private.
\end{lemma}

\subsection{Utility Analysis}
\label{sec:utility_analysis}

Given a database $S\subseteq X$ over a totally ordered domain $X$ and a point $y\in X$, we say that $y$ is an interior point of $S$ with score $\Delta$ if there are at least $\Delta$ elements in $S$ that are greater or equal to $y$ and there are at least $\Delta$ elements in $S$ that are smaller or equal to $y$. We say that $y$ is an interior point, without mentioning the score, to mean that it is an interior point with score at least 1.

\begin{lemma}\label{lem:TreeLogUtility}
Let algorithm \texttt{TreeLog} be executed on a database $S\in X^n$ of size $n\geq O\left(\frac{\log^*|X|}{\eps}\log\frac{1}{\delta}\right)$. With probability at least $1-O\left(\delta\log^*|X|\right)$, the algorithm returns an interior point of $S$ with score at least $\Omega\left(\frac{1}{\eps}\log\frac{1}{\delta}\right)$.
\end{lemma}

\begin{proof}
The utility analysis is by induction on the number of recursive calls. Specifically, let $N$ denote the number of recursive calls, and let $S_i$ denote the input to the $i$th recursive call (we identify the top-level call with the index $0$, so $S_0=S$, and the deepest call with the index $N$). We denote $n_i=|S_i|$. We will show that if the $(i+1)$th call returns an interior point of $S_{i+1}$ then, w.h.p., so does the $i$th call.

For the base case, observe that every iteration of \texttt{TreeLog} shrinks the size of the domain logarithmically, and hence,  $N=O(\log^*|X|)$. In contrast, the database size only decreases additively in each recursive call, by a factor of $3t$. Therefore, when $n\geq O(t\cdot \log^*|X|)=O(\frac{\log^*|X|}{\eps}\log\frac{1}{\delta})$ the last recursive call (which halts in Step~1) is performed with a database of size $\Omega(\frac{1}{\eps}\log\frac{1}{\delta})$. In this case, by
Lemma \ref{prop:exp_mech} (Exponential mechanism with $\beta=\delta$), the last recursive call returns an interior point of $S_N$ with score at least $\Omega\left(\frac{1}{\eps}\log\frac{1}{\delta}\right)$ with probability at least $(1-\delta)$.

Now  consider the $i$th call. Let $\hat{S}_i$ be the database of size $n_i-2t$ constructed in Step~\ref{step:takeMiddle}, let $\pi$ be the path selected in Step~\ref{step:randomPath}, and let $D$ be the database constructed in Step~\ref{step:buildD}. Assume that by induction, with probability at least $(1-2\delta(N-(i+1)))$, the recursive call in Step~\ref{step:recursiveCall} returned a point $\ell^*$ that is an interior point of $D$ with score at least $\Omega\left(\frac{1}{\eps}\log\frac{1}{\delta}\right)$. This means that at least $\Omega\left(\frac{1}{\eps}\log\frac{1}{\delta}\right)$ points from $\hat{S}_i$ fall off the path $\pi$ on or before level $\ell^*$ of the tree $T$, and at least $\Omega\left(\frac{1}{\eps}\log\frac{1}{\delta}\right)$ points from $\hat{S}_i$ fall off on or after level $\ell^*$.

Let $v_{\pi}^{\ell^*}$ denote the node in $\pi$ at level $\ell^*$ of the tree. Since at least $\Omega\left(\frac{1}{\eps}\log\frac{1}{\delta}\right)$ points from $\hat{S}_i$ fall off from $\pi$ on or after level $\ell^*$, we get that the weight of $v_{\pi}^{\ell^*}$ is at least $\Omega\left(\frac{1}{\eps}\log\frac{1}{\delta}\right)$. Hence, by
Lemma \ref{lem:CM} (Choosing mechanism with $\beta=\delta$), the node $v^*$ identified in Step~\ref{step:choosingMech} also has weight at least $\Omega\left(\frac{1}{\eps}\log\frac{1}{\delta}\right)$ with probability at least $(1-\delta)$. Suppose that the weight of $v^*$ is strictly less than $|\hat{S}_i|$. Then either $v_{\scriptscriptstyle\rm left}$ or $v_{\scriptscriptstyle\rm right}$ is an interior point of $\hat{S}_i$, and hence, it is an interior point of $S_i$ with score at least $\Omega\left(\frac{1}{\eps}\log\frac{1}{\delta}\right)$. In this case, with probability at least $(1-\delta)$, the Exponential Mechanism in Step~\ref{step:expMech} identifies an interior point of $S_i$ with score at least $\Omega\left(\frac{1}{\eps}\log\frac{1}{\delta}\right)$.

Now suppose that $v^*$ has weight exactly $|\hat{S}_i|$, which means that $v^*=v_{\pi}^{\ell^*}$ (since in this case the path $\pi$ must go through $v^*$). Recall that at least $\Omega\left(\frac{1}{\eps}\log\frac{1}{\delta}\right)$ points from $\hat{S}_i$ fall off from $\pi$ on or before level $\ell^*$. Since $v^*=v_{\pi}^{\ell^*}$ has weight $|\hat{S}_i|$, this means that all these points fall of $\pi$ exactly on level $\ell^*$. So both children of $v^*$ have positive weight. This means that either $v_{\scriptscriptstyle\rm inner\text{-}left}$ of $v_{\scriptscriptstyle\rm inner\text{-}right}$ is an interior point of $\hat{S}_i$, and hence, it is an interior point of $S_i$ with score at least $\Omega\left(\frac{1}{\eps}\log\frac{1}{\delta}\right)$. As we argued in the base case of the induction, with probability at least $(1-\delta)$, the Exponential Mechanism in Step~\ref{step:expMech} identifies an interior point of $S_i$ with score at least $\Omega\left(\frac{1}{\eps}\log\frac{1}{\delta}\right)$.

Overall, with probability at least $(1-2\delta(N-i))$, the $i$th call returns an interior point of $S_i$ with score at least $\Omega\left(\frac{1}{\eps}\log\frac{1}{\delta}\right)$.
\end{proof}

Theorem~\ref{thm:introPowerTwo} now follows by combining Lemma~\ref{lem:TreeLogPrivacy} and Lemma~\ref{lem:TreeLogUtility}. 

\section{Decomposing Algorithm \texttt{TreeLog}}\label{sec:decompose}

In this section we reduce the sample complexity of algorithm \texttt{TreeLog} from $\tilde{O}\left(\left(\log^*|X|\right)^2\right)$ to $\tilde{O}\left(\left(\log^*|X|\right)^{1.5}\right)$. Specifically, we show the following theorem.

\begin{theorem}\label{thm:mainAdvanced}
Let $X$ be a totally ordered domain. There exists an $(\eps,\delta)$-differentially private algorithm for solving the interior point problem on databases of size $n=\Tilde{O}\left(\frac{1}{\eps}\cdot\left(\log^*|X|\right)^{1.5}\cdot\log^{1.5}(\frac{1}{\delta})\right)$.
\end{theorem}

Recall that in the construction in Section~\ref{sec:TreeLog}, the database size is reduced (additively) by $\approx\frac{1}{\eps}\log\frac{1}{\delta}$ in each iteration, and hence we needed to start with at least $\approx\frac{\log^*|X|}{\eps}\log\frac{1}{\delta}$ input points in order to make it to the end of the recursion without losing all the data points. In addition, with these parameters, algorithm \texttt{TreeLog} only guaranteed $\approx(\eps\cdot\log^*|X|,\delta\cdot\log^*|X|)$-differential privacy. To get $(\eps,\delta)$-privacy overall, we divided our privacy parameters by $\approx\frac{1}{\log^*|X|}$, which gives sample complexity of $\tilde{O}\left(\left(\log^*|X|\right)^2\right)$. Informally, in this section we apply composition theorems for differential privacy to argue that it suffices to work with a privacy parameter of $\approx\frac{\eps}{\sqrt{\log^*|X|}}$, which would result in a sample complexity of $\tilde{O}\left(\left(\log^*|X|\right)^{1.5}\right)$. However, as we explain next, this requires some additional technical work, and does not follow from a direct application of existing composition theorems to the construction of Section~\ref{sec:TreeLog}.

Recall that composition theorems for differential privacy state that the application of $k$, $(\eps,\delta)$-differentially private mechanisms,  satisfies $\approx(\eps\sqrt{k},\delta k)$-differential privacy. %
For example, consider an algorithm $\BBB$ that interacts with its input database $S$ only through differentially private mechanisms, as follows. In every step $i\in[k]$, algorithm $\BBB$ selects an $(\eps,\delta)$-differentially private mechanism $\AAA_i$, and runs $\AAA_i$ on $S$ to obtain an outcome $a_i$, where the choice of $\AAA_i$ might depend on the previous outcomes $a_1,\dots,a_{i-1}$. Assuming that this is the only interaction $\BBB$ has with its input database, then these theorems state that $\BBB$ is $\approx(\eps\sqrt{k},\delta k)$-differentially private.

Unraveling the recursion in
Algorithm \texttt{TreeLog} from 
Section~\ref{sec:TreeLog}, we observe that it consists  of
$2\log^*|X|$ steps. Each of the first 
$\log^*|X|$ steps computes the input to the following recursive call, and each of the last 
$\log^*|X|$ steps uses the output of the recursive call to compute an interior point of
large score.
The difficulty in applying the standard composition arguments  to \texttt{TreeLog} 
arises since the first $\log^*|X|$ steps are obviously not differentially private: They just spit out the private input compressed into a smaller domain. (Furthermore, the path which we use to do the compression is also highly sensitive.) 

Alternatively, one can think of Algorithm \texttt{TreeLog} as consisting of $\log^*|X|$ steps, where the first step computes all of the paths and databases throughout the execution down to the last recursive call, and then privately identifies an interior point of the database in that last recursive call. %
Afterwards, each of the next steps takes a privately computed interior point for the database at depth $d$ of the recursion, and privately translates it into an interior point for the database at depth $d-1$. Now, when viewing algorithm \texttt{TreeLog} this way, the difficulty in applying composition theorems is that all of these $\log^*|X|$ steps {\em share a state} which was {\em not} computed privately. Specifically, all of these steps know (and use) the paths (or the databases) that were computed throughout the execution. This picture is not supported by existing composition theorems, which allow only a {\em privately} computed state to be shared among the composed mechanisms.

Loosely speaking, we overcome this difficulty by modifying
\texttt{TreeLog} such that 
it becomes a composition of $O(\log^*|X|)$
differentially private mechanisms that do not share a non-private state. 
The crucial observation that enables this modification is that once we reach a level $d^*\in[\log^*|X|]$ of the recursion in which many points fall off the heavy path at the same node, then we can in fact stop the recursion and {\em privately} report the level $d^*$. 
Recall that in algorithm \texttt{TreeLog} we sample a random path $\pi$ by proceeding randomly from a node $v$ to one of its children with probability that grows exponentially with the weight of the child. Now, since $d^*$ is the first level in the recursion in which many points fall off the heavy path at the same node, in all levels $d<d^*$, when constructing the path $\pi$ we always have that the weight of one child is significantly bigger then the weight of its sibling. This means that we will in fact proceed to the heavier child with overwhelming probability, to the extent that we can ignore the randomness in that step and just select it deterministically.

This allows us to modify the way in which the paths are selected throughout the execution to be {\em deterministic}. As a result, we are able to
 decompose  \texttt{TreeLog} into $O(\log^*|X|)$ different algorithms (one for every level of the recursion), where each of these algorithms recomputes from scratch all of the input databases up to its current level. More specifically, we decompose algorithm \texttt{TreeLog} into the following 5 algorithms:

\begin{enumerate}
	\item[{\bf 1.}] {\bf Algorithm \texttt{ConstructPaths}:} This algorithm (deterministically) generates all the paths and the databases for the entire execution. The other algorithms use \texttt{ConstructPaths} as a subroutine in order to recompute the paths and the databases.
	\item[{\bf 2.}] {\bf Algorithm \texttt{StoppingPoint}:} This algorithm instantiates the sparse vector technique in order to privately identify the first level in which, when following the heavy path deterministically, we reach a node such that both its children have large weight. In other words, this algorithm privately computes a stopping point for the construction.
	\item[{\bf 3.}] {\bf Algorithm \texttt{OneRandomPath}:} This algorithm is similar to (one call of) algorithm \texttt{TreeLog} that samples a path randomly. It is applied only once at the last level in the execution, where we might not be able to select the path deterministically.
	\item[{\bf 4.}] {\bf Algorithm \texttt{LevelUp}:} This algorithm gets a parameter $d$ and a point $y$ which is an interior point of the $(d+1)$th database that algorithm \texttt{ConstructPaths} generates, and translates $y$ into an interior point of the $d$th database that \texttt{ConstructPaths} generates. This algorithm is executed up to $\log^*|X|$ times (with different parameters).
		\item[{\bf 5.}] {\bf Algorithm \texttt{HeavyPaths}:} This is a wrapper algorithm that runs the previous 4 algorithms.
\end{enumerate}

Before presenting these five algorithms, we recall the sparse vector technique of Dwork et al.~\cite{DNRRV09}, which we use in order to compute a stopping point for the construction. Consider a sequence of low sensitivity functions $f_1,f_2,\ldots,f_m$, which are given (one by one) to a data curator (holding a database $S$). Algorithm \texttt{AboveThreshold} by Dwork et al.~\cite{DNRRV09} privately identifies the first query $f_i$ whose value $f_i(S)$ is greater than some threshold. Formally,

\begin{theorem}[Algorithm \texttt{AboveThreshold}~\cite{DNRRV09}]\label{thm:AboveThreshold}
There exists an $(\eps,0)$-differentially private algorithm $\cal A$ such that for $m$ rounds, after receiving a sensitivity-1 query $f_i:U^*\rightarrow\R$, algorithm $\cal A$ either outputs $\top$ and halts, or outputs $\bot$ and waits for the next round.
If $\cal A$ was executed with a database $S\in U^*$ and a threshold parameter $c$, then the following holds with probability $(1-\beta)$:
(i) If a query $f_i$ was answered by $\top$ then $f_i(S)\geq c-\frac{8}{\eps}\log(2m/\beta)$;
(ii) If a query $f_i$ was answered by $\bot$ then $f_i(S)\leq c+\frac{8}{\eps}\log(2m/\beta)$.
\end{theorem}

We are now ready to present our improved construction for privately identifying an interior point. 
We begin by presenting algorithms \texttt{ConstructPaths} and \texttt{StoppingPoint}. For their analysis, we need the following notation.
Let $\lambda$ be a global constant, and
let 
$k=\frac{\lambda}{\eps}\cdot\log\left(\frac{1}{\delta}\cdot\log^*|X|\right)$.
 Consider the functions $f_{d}$ defined in algorithm \texttt{StoppingPoint}. For a database $S\in X^n$ we define $d^*(S)$ to be the smallest index such that $f_{d^*(S)}(S)\geq k$.

\begin{algorithm*}[!htp]

\caption{\bf \texttt{ConstructPaths}}\label{alg:ConstructPaths}

{\bf Input:} A database $S\in X^n$ where $X$ is a totally ordered domain. We assume that $|X|$ is a power of 2 (otherwise extend $X$ to the closest power of 2).

{\bf Global parameter:} $t$.

\begin{enumerate}[leftmargin=15pt,rightmargin=10pt,itemsep=1pt,topsep=1.5pt]

\item\label{ConstructPaths:end} If $|S|\leq10t$ or if $|X|=O(1)$ then halt. Otherwise continue to the next step.

\item\label{ConstructPaths:takeMiddle} Sort $S$ and let $\hat{S}\in X^{n-2t}$ be a database containing all elements of $S$ except for the $t$ largest and $t$ smallest elements.

\item\label{ConstructPaths:buildTree} Let $T$ be a complete binary tree with $|X|$ leaves that correspond to elements of $X$. A leaf $u$ has weight $w(u)=|\{x\in \hat{S}: x=u\}|$. A node $v$ has weight $w(v)$ that equals the sum of the weights of its children.

\item\label{ConstructPaths:randomPath} Let $\pi$ be the path in $T$ constructed as follows (starting from the root):
\begin{enumerate}[topsep=0pt]
	\item Let $v$ be the current node in the path.
	\item If $v$ is a leaf, or if $w(v)\leq t$, then $v$ is the last node in $\pi$.
	\item\label{ConstructPaths:randomPathProceed} Else, let $v_0,v_1$ be the two children of $v$ in $T$. Proceed to the child $v_b$ with larger weight $w(v_b)$, where $b\in\{0,1\}$ and where ties are broken arbitrarily.
\end{enumerate}

\item\label{ConstructPaths:buildD} Initialize $D=\emptyset$. Add elements to $D$ by following the path $\pi$ starting from the root:
\begin{enumerate}[topsep=0pt]
	\item\label{ConstructPaths:buildDa} If $|D|\geq n-3t$ then goto Step~\ref{ConstructPaths:output}.
	\item Let $v$ be the current node in the path, and let $\ell$ denote its level in $T$ (the root is in level 0 and the leaves are in level $\log|X|$).
	\item If $v$ is the last node in $\pi$ then add $(n-3t-|D|)$ copies of $\ell$ to $D$ and goto Step~\ref{ConstructPaths:output}.
	\item Else, let $v_{\scriptscriptstyle\rm next}$ be the next node in $\pi$, and let $v_{\scriptscriptstyle\rm other}$ be the other child of $v$ in $T$. Add $\min\left\{w(v_{\scriptscriptstyle\rm other}),n-3t-|D|\right\}$ copies of $\ell$ to $D$, and goto Step~\ref{ConstructPaths:buildDa} with $v_{\scriptscriptstyle\rm next}$ as the current node.
\end{enumerate}

\item\label{ConstructPaths:output} Output the domain $X$, the database $S$, the tree $T$, and the path $\pi$.

\item\label{ConstructPaths:recursiveCall} Execute \texttt{ConstructPaths} recursively on $D$.

\end{enumerate}
\end{algorithm*}

\begin{algorithm*}[!htp]

\caption{\bf \texttt{StoppingPoint}}\label{alg:StoppingPoint}

{\bf Input:} Parameters $\eps,\delta$ and a database $S\in X^n$ where $X$ is a totally ordered domain.

{\bf Additional input:} Parameter $t=2k=\frac{2\lambda}{\eps}\cdot\log\left(\frac{1}{\delta}\cdot\log^*|X|\right)$, where $\lambda$ is a global constant.

\begin{enumerate}[leftmargin=15pt,rightmargin=10pt,itemsep=1pt,topsep=1.5pt]

\item\label{StoppingPoint:AboveThreshInit} Instantiate algorithm \texttt{AboveThreshold} (see Theorem~\ref{thm:AboveThreshold}) with the database $S$, privacy parameter $\eps$, and threshold $c=k/2$.

\item\label{StoppingPoint:loop} For $d=1,2,\dots,\log^*|X|$ do

\begin{enumerate}[topsep=0pt]
	\item Define the following query $f_{d}:X^*\rightarrow\N$. To compute $f_{d}$ on a database $S$, apply algorithm \texttt{ConstructPaths} on $S$ with parameter $t=4c$, and let $S_{d}\in(X_{d})^*$ be the $d$th database that it outputs. Then $f_{d}(S)=\max_{y\in X_{d}}\left|\{x\in S_{d} : x=y\}\right|$.

	\item Query algorithm \texttt{AboveThreshold} on $f_{d}$. If the answer is $\top$ then halt and output $d$. Otherwise continue.
\end{enumerate}

\end{enumerate}
\end{algorithm*}

In the next lemma we analyze how \texttt{ConstructPaths} behaves on neighboring databases.

\begin{lemma}\label{lem:ConstructPathsNeighboring}
Let $S\in X^n$ and $S'\in X^n$ be two neighboring databases, and
let $t\geq2k=\frac{2\lambda}{\eps}\cdot\log\left(\frac{1}{\delta}\cdot\log^*|X|\right)$. Let $\{(X_{d},S_{d},T_{d},\pi_{d})\}_{d=1}^{\log^*|X|}$ and $\{(X_{d},S'_{d},T_{d},\pi'_{d})\}_{d=1}^{\log^*|X|}$ denote the outcomes of the executions of \texttt{ConstructPaths} with parameter $t$ on $S$ and on $S'$, respectively. Then for every $d<d^*(S)$ we have that $S_{d}$ and $S'_{d}$ are neighboring databases.
\end{lemma}

\begin{proof}
The proof is by induction on $d$. For the base case, observe that $S_1=S$ and $S'_1=S'$ are neighboring databases. Now fix $d<d^*(S)$, and suppose that $S_{d-1}$ and $S'_{d-1}$ are neighboring databases. Since $d$ is strictly smaller than $d^*(S)$, we have that all the elements of $S_{d}$ have multiplicities less than $k\triangleq\frac{\lambda}{\eps}\cdot\log\left(\frac{1}{\delta}\cdot\log^*|X|\right)$. By the way $S_{d}$ is constructed, this means that less than $k$ elements from $S_{d-1}$ fall off the path $\pi_{d-1}$ in every single node of $T_{d-1}$. Recall that when constructing the path $\pi_{d-1}$, every node in the path has weight at least $t$ (except maybe for the last node in the path whose weight might be $t/2$). This means that throughout the construction of $\pi_{d-1}$ in Step~\ref{ConstructPaths:randomPath} of \texttt{ConstructPaths}, the weight of the child we proceed to (in Step~\ref{ConstructPaths:randomPathProceed}) is always larger than the weight of its sibling by  at least $t-k\gg1$. Since $S_{d-1}$ and $S'_{d-1}$ are neighboring databases by the induction assumption, this gap is also bigger than 1 when constructing $\pi'_{d-1}$ during the execution of $S'$. As a result, the construction of $\pi'_{d-1}$ proceeds identically to the construction of $\pi_{d-1}$, except possibly that one path might be longer than the other. So $\pi_{d-1}$ and $\pi'_{d-1}$ are either exactly the same path, or one of them is a prefix of the other. As in the analysis of algorithm \texttt{TreeLog} in Section~\ref{sec:TreeLog}, in such a case we have that the resulting $S_{d}$ and $S'_{d}$ are neighboring databases.
\end{proof}

The next lemma specifies the utility  of algorithm \texttt{StoppingPoint}.

\begin{lemma}\label{lem:StoppingPointUtility}
Let $S\in X^n$ be a database
of size $n=\Omega\left(\frac{\log^*|X|}{\eps}\cdot\log\left(\frac{1}{\delta}\cdot\log^*|X|\right)\right)$, and let $\hat{d}$ denote the outcome of \texttt{StoppingPoint} on $S$. Then, with probability at least $1-\delta$ we have that $\hat{d}\leq d^*(S)$, and that $f_{\hat{d}}(S)\geq\Omega(k)=\Omega\left(\frac{1}{\eps}\cdot\log\left(\frac{1}{\delta}\cdot\log^*|X|\right)\right)$.
\end{lemma}

Lemma~\ref{lem:StoppingPointUtility} follows directly from the utility properties of algorithm \texttt{AboveThreshold} (see Theorem~\ref{thm:AboveThreshold}; use $\beta=\delta$, notice that the additive error of \texttt{AboveThreshold} is at most $k/2$ for an appropriate choice of $\lambda$), after observing that when $n=\Omega\left(\frac{\log^*|X|}{\eps}\cdot\log\left(\frac{1}{\delta}\cdot\log^*|X|\right)\right)$ there must exist an index $d$ such that $f_{d}(S)\geq k$  and therefore $d^*(S)$ is well defined. This follows since after $\log^*|X|$ iterations of algorithm \texttt{ConstructPaths} we get that the size of the domain is constant, while the size of the database is still at least $n-O(t\log^*|X|)$. We now proceed with the privacy analysis of Algorithm \texttt{StoppingPoint}.

\begin{lemma}\label{lem:StoppingPointPrivacy}
Algorithm \texttt{StoppingPoint} is $(\eps,\delta)$-differentially private.
\end{lemma}

\begin{proof}
Observe that the outcome of algorithm \texttt{StoppingPoint} is a post-processing of the outcomes of algorithm \texttt{AboveThreshold}. Hence, it suffices to argue that the sequence of outcomes obtained from algorithm \texttt{AboveThreshold} satisfies $(\eps,\delta)$-differential privacy. Intuitively, but somewhat inaccurately, this will be done by showing that with probability $(1-\delta)$ all the queries issued to algorithm \texttt{AboveThreshold} are of sensitivity 1, in which case algorithm \texttt{AboveThreshold} guarantees $(\eps,0)$-differential privacy.

Let $S\in X^n$ and $S'\in X^n$ be two neighboring databases, and let $\{(X_{d},S_{d},T_{d},\pi_{d})\}_{d=1}^{\log^*|X|}$ and $\{(X_{d},S'_{d},T_{d},\pi'_{d})\}_{d=1}^{\log^*|X|}$ denote the outcomes of the executions of \texttt{ConstructPaths} on $S$ and on $S'$, respectively. By Lemma~\ref{lem:ConstructPathsNeighboring}, for every $d<d^*(S)$ it holds that $S_{d}$ and $S'_{d}$ are neighboring databases. Therefore, for every $d<d^*(S)$, we have that $|f_{d}(S)-f_{d}(S')|\leq1$ by the definition of the query $f_{d}$.

We now argue that also for $d\triangleq d^*(S)$ we have that $|f_{d}(S)-f_{d}(S')|\leq1$, even though $S_{d}$ and $S'_{d}$ might not be neighboring databases. Recall that $f_{d}(S)$ equals to the maximal multiplicity of an element of $S_{d}$. Equivalently, $f_{d}(S)$ is determined by the weight (w.r.t.\ $S_{d-1}$) of the heaviest node that falls off the previous path $\pi_{d-1}$.
Suppose that $\pi_{d-1}$ and $\pi'_{d-1}$ are {\em not} the same path and that none of them is a prefix of the other (as otherwise $S_{d}$ and $S'_{d}$ would be neighboring databases and then $|f_{d}(S)-f_{d}(S')|\leq 1$). Let $u$ be the last node common to  $\pi_{d-1}$ and $\pi'_{d-1}$, and let $v_0,v_1$ denote its children. 
Since $S_{d-1}$ and $S'_{d-1}$ are neighboring, we have that $|w(v_0)-w'(v_0)|\leq 1$ and that $|w(v_1)-w'(v_1)|\leq 1$, where here $w$ and $w'$ denote the weight of the elements of $T_{d-1}$ during the executions on $S$ and on $S'$, respectively. In addition, since one path proceeds to $v_0$ and the other proceeds to $v_1$, we have that
$|w(v_0)-w(v_1)|\leq2$ and that $|w'(v_0)-w'(v_1)|\leq2$. 
Note that $z\triangleq\min\{w(v_0),w(v_1)\}$ elements fall off the path $\pi_{d-1}$ at the node $v$, and that less than $z$ elements can fall off this path after node $v$, since we would be left with only $\max\{w(v_0),w(v_1)\}\leq z+2$ potential elements that can fall off the path after $v$ (which means that the maximal fall could be of size at most $(z+2)/2$). Hence, $f_{d}(S)$ and $f_{d}(S')$ are determined by the weight (w.r.t.\ $S_{d-1}$ and $S'_{d-1}$) of the heaviest nodes that fall off the paths $\pi_{d-1}$ and $\pi'_{d-1}$, {\em before or at} the node $v$. Since the paths are the same until $v$, and since $S_{d-1}$ and $S'_{d-1}$ are neighboring databases, we get that $|f_{d}(S)-f_{d}(S')|\leq 1$.

So, for the first $d^*(S)$ queries that are issued to algorithm \texttt{AboveThreshold} we have that $|f_{d}(S)-f_{d}(S')|\leq1$. Moreover, by Lemma~\ref{lem:StoppingPointUtility},  with probability $1-\delta$, other queries are never issued to algorithm \texttt{AboveThreshold}. The lemma therefore follows from the privacy properties of \texttt{AboveThreshold}. Specifically, let $\tau$ and $\tau'$ denote the sequence of outcomes of \texttt{AboveThreshold} during the executions of algorithm \texttt{StoppingPoint} on $S$ and on $S'$, respectively. Let $T$ be some set of  possible such sequences,  let $T_{\rm long}\subset T$ contain all sequences longer than $d^*(S)$, and let $T_{\rm short}=T\setminus T_{\rm long}$. We have that,
\begin{align*}
\Pr[\tau\in T]&=\Pr[\tau\in T_{\rm short}] + \Pr[\tau\in T_{\rm long}]\\
&\leq\e^{\eps}\cdot\Pr[\tau'\in T_{\rm short}]+\delta\\
&\leq\e^{\eps}\cdot\Pr[\tau'\in T]+\delta.
\end{align*}
\end{proof}

\begin{algorithm*}[!htp]

\caption{\bf \texttt{OneRandomPath}}\label{alg:OneRandomPath}

{\bf Input:} Parameters $\eps,\delta$ and a database $S\in X^n$ where $X$ is a totally ordered domain. We assume that $|X|$ is a power of 2 (otherwise extend $X$ to the closest power of 2).

{\bf Additional inputs:} Parameter $t=2k=\frac{2\lambda}{\eps}\cdot\log\left(\frac{1}{\delta}\cdot\log^*|X|\right)$, where $\lambda$ is a global constant, current level $d\in[\log^*|X|]$.

\begin{enumerate}[leftmargin=15pt,rightmargin=10pt,itemsep=1pt,topsep=1.5pt]

\item Apply algorithm \texttt{ConstructPaths} on $S$ and let $(X_{d},S_{d},T_{d},\pi_{d})$ denote its $d$th outputs.

\item\label{step:takeMiddle_OneRandomPath} Sort $S_{d}$ and let $\hat{S}_{d}\in X^{n-2t}$ be a database containing all elements of $S_{d}$ except for the $t$ largest and $t$ smallest elements.

\item\label{step:buildTree_OneRandomPath} Let $T$ be a complete binary tree with $|X|$ leaves that correspond to elements of $X$. A leaf $u$ has weight $w(u)=|\{x\in \hat{S}_{d}: x=u\}|$. A node $v$ has weight $w(v)$ that equals the sum of the weights of its children.

\item\label{step:randomPath_OneRandomPath} Sample a path $\pi$ in $T$, starting from the root and constructed by the following process:
\begin{enumerate}[topsep=0pt]
	\item Let $v$ be the current node in $\pi$.
	\item If $v$ is a leaf, or if $w(v)\leq t$, then $v$ is the last node in $\pi$.
	\item Else, let $v_0,v_1$ be the two children of $v$ in $T$. If one of them has weight 0 then proceed to the other child. Otherwise, proceed to $v_b$ (for $b\in\{0,1\}$) with probability proportional to $\exp(\eps\cdot w(v_b))$.
\end{enumerate}

\item\label{step:buildD_OneRandomPath} Initialize $D=\emptyset$. Add elements to $D$ by following the path $\pi$ starting from the root:
\begin{enumerate}[topsep=0pt]
	\item\label{step:buildDa} If $|D|\geq n-3t$ then goto Step~\ref{step:recursiveCall}.
	\item Let $v$ be the current node in the path, and let $\ell$ denote its level in $T$ (the root is in level 0 and the leaves are in level $\log|X|$).
	\item If $v$ is the last node in $\pi$ then add $(n-3t-|D|)$ copies of $\ell$ to $D$ and goto Step~\ref{step:recursiveCall}.
	\item Else, let $v_{\scriptscriptstyle\rm next}$ be the next node in $\pi$, and let $v_{\scriptscriptstyle\rm other}$ be the other child of $v$ in $T$. Add $\min\left\{w(v_{\scriptscriptstyle\rm other}),n-3t-|D|\right\}$ copies of $\ell$ to $D$, and goto Step~\ref{step:buildDa} with $v_{\scriptscriptstyle\rm next}$ as the current node.
\end{enumerate}

\item\label{step:choosingMech_OneRandomPath} Use the Choosing Mechanism to choose an element $i\in[\log|X|]$ with large multiplicity in $D$.

\item Use the Choosing Mechanism to choose a node $v^*$ at level $i$ of $T_{d}$ with large weight $w(v^*)$.

\item Let $v_{\scriptscriptstyle\rm left}$ and $v_{\scriptscriptstyle\rm right}$ be the left-most and right-most leaves, respectively, of the sub-tree rooted at $v^*$. Also let $v_{\scriptscriptstyle\rm inner\text{-}left}$ be the right-most leaf of the sub-tree rooted at the left child of $v^*$, and let $v_{\scriptscriptstyle\rm inner\text{-}right}$ be the left-most leaf of the sub-tree rooted at the right child of $v^*$.

\item Use the Exponential Mechanism to return $y\in\left\{v_{\scriptscriptstyle\rm left}, v_{\scriptscriptstyle\rm right}, v_{\scriptscriptstyle\rm inner\text{-}left}, v_{\scriptscriptstyle\rm inner\text{-}right} \right\}$ with large quality
$$q(S_{d},y)=\min\left\{\;|\{x\in S_{d}: x\leq y\}|,\; |\{x\in S_{d}: x\geq y\}|\;\right\}.$$

\end{enumerate}
\end{algorithm*}

We now present the analysis of algorithm \texttt{OneRandomPath}.

\begin{lemma}\label{lem:OneRandomPath_Privacy}
Fix two neighboring databases $S,S'\in X^n$, and consider running algorithm \texttt{OneRandomPath} on $S$ and on $S'$ with parameter $d<d^*(S)$. Then for every set of outcomes $T$ we have
$$
\Pr[\texttt{OneRandomPath}(S)\in T]\leq e^{6\eps\cdot\log n}\cdot\Pr[\texttt{OneRandomPath}(S')\in T]+4\delta n\cdot e^{4\eps\cdot\log n}.
$$
\end{lemma}

The proof of Lemma~\ref{lem:OneRandomPath_Privacy} is essentially identical to the privacy analysis of (a single iteration of) algorithm \texttt{TreeLog} (after applying Lemma~\ref{lem:ConstructPathsNeighboring} to argue that in the executions on $S$ and on $S'$ we have that the resulting $S_{d}$ and $S'_{d}$ are neighboring databases). We omit the proof for brevity (see Lemmas~\ref{lem:TreeLogPrivacyInduction} and~\ref{lem:TreeLogPrivacy}).
(The constants are different than in Lemma~\ref{lem:TreeLogPrivacy} since \texttt{OneRandomPath} applies the Choosing Mechanism twice where \texttt{TreeLog} applies it only once.)

\begin{lemma}\label{lem:OneRandomPath_Utility}
Let $t=2k=\frac{2\lambda}{\eps}\cdot\log\left(\frac{1}{\delta}\cdot\log^*|X|\right)$, where $\lambda$ is a global constant, and let $S\in X^n$ be a database.
Let $\{(X_{d},S_{d},T_{d},\pi_{d})\}_{d=1}^{\log^*|X|}$ denote the outcomes of the executions of \texttt{ConstructPaths} with parameter $t$ on $S$. Let algorithm \texttt{OneRandomPath} be executed on $S$ with parameters $t,d$ such that $f_{d+1}(S)\geq\Omega\left(\frac{1}{\eps}\log\frac{1}{\delta}\right)$, 
where the function $f_{d+1}(\cdot)$ is defined in algorithm \texttt{StoppingPoint}. Then, with probability at least $1-\delta$, the outcome $y$ is an interior point of $S_{d}$ with depth at least $\Omega(t)$. 
\end{lemma}

\begin{proof}[Proof Sketch]
The proof is almost identical to the utility analysis of algorithm \texttt{TreeLog}, with the following exception. As we next explain, the assumption that $f_{d+1}(S)\geq\Omega\left(\frac{1}{\eps}\log\frac{1}{\delta}\right)$ %
means that in the database $D$ (constructed in Step~4) there is at least one element with large multiplicity. Hence, instead of applying the recursion on $D$ (as we did in algorithm \texttt{TreeLog}), we could use the Choosing Mechanism (Lemma \ref{lem:CM} with confidence $\beta =\delta$) in order to identify such an element, which would be an interior point of $D$ with large score.

The assumption that $f_{d+1}(S)\geq\Omega\left(\frac{1}{\eps}\log\frac{1}{\delta}\right)$ %
means that if we were to generate the path $\pi$ {\em deterministically} (as is done in algorithm \texttt{ConstructPaths}) then the resulting database $D$ would be guaranteed to contain at least one element with large multiplicity. Let $\pi_{\rm deter}$ denote the path that would be obtained by following the heaviest nodes deterministically. We have established that if $\pi=\pi_{\rm deter}$, then $D$ contains an element with large multiplicity. Otherwise, if $\pi\neq\pi_{\rm deter}$, then $\pi$ ``breaks off'' from the heaviest path at some node $v$. Since the path $\pi_{\rm deter}$ proceeds to a different child of $v$ than $\pi$, we have that ``a lot'' of elements fall off of the path $\pi$ at the node $v$. Therefore, in that case we again get that the resulting database $D$ contains at least one element with large multiplicity.
\end{proof}

\begin{algorithm*}[!htp]

\caption{\bf \texttt{LevelUp}}\label{alg:LevelUp}

{\bf Input:} Parameters $\eps,\delta$ and a database $S\in X^n$ where $X$ is a totally ordered domain. We assume that $|X|$ is a power of 2.

{\bf Additional inputs:} Parameter $t=2k=\frac{2\lambda}{\eps}\cdot\log\left(\frac{1}{\delta}\cdot\log^*|X|\right)$, where $\lambda$ is a global constant, the current recursive level $d\in[\log^*|X|]$, and a level $i$ in the tree $T_d$ (which is an interior point in $S_{d+1}$).

\begin{enumerate}[leftmargin=15pt,rightmargin=10pt,itemsep=1pt,topsep=1.5pt]

\item Apply algorithm \texttt{ConstructPaths} on $S$ with the parameter $t$ and let $(X_{d},S_{d},T_{d},\pi_{d})$ denote its $d$th output.

\item Sort $S_{d}$ and let $\hat{S}_{d}$ be a database containing all elements of $S_{d}$ except for the $t$ largest and $t$ smallest elements.

\item A leaf $u$ of the tree $T_{d}$ has weight $w(u)=|\{x\in \hat{S_{d}}: x=u\}|$. A node $v$ has weight $w(v)$ that equals the sum of the weights of its children. Use the Choosing Mechanism to choose a node $v^*$ at level $i$ of $T_{d}$ with large weight $w(v^*)$.

\item Let $v_{\scriptscriptstyle\rm left}$ and $v_{\scriptscriptstyle\rm right}$ be the left-most and right-most leaves, respectively, of the sub-tree rooted at $v^*$. Also let $v_{\scriptscriptstyle\rm inner\text{-}left}$ be the right-most leaf of the sub-tree rooted at the left child of $v^*$, and let $v_{\scriptscriptstyle\rm inner\text{-}right}$ be the left-most leaf of the sub-tree rooted at the right child of $v^*$.

\item Use the Exponential Mechanism to return $y\in\left\{v_{\scriptscriptstyle\rm left}, v_{\scriptscriptstyle\rm right}, v_{\scriptscriptstyle\rm inner\text{-}left}, v_{\scriptscriptstyle\rm inner\text{-}right} \right\}$ with large quality
$$q(S_{d},y)=\min\left\{\;|\{x\in S_{d}: x\leq y\}|,\; |\{x\in S_{d}: x\geq y\}|\;\right\}.$$

\end{enumerate}
\end{algorithm*}

We now proceed with the analysis of algorithm \texttt{LevelUp}. Recall that this algorithm gets a recursive level $d$ and a level $i$
in $T_d$ which is an interior point of the $(d+1)$th database that algorithm \texttt{ConstructPaths} generates, and translates $i$ into an interior point of the $d$th database that \texttt{ConstructPaths} generates. The following lemma specifies the utility properties of the algorithm.

\begin{lemma}\label{lem:LevelUpUtility}
Let $t=2k=\frac{2\lambda}{\eps}\cdot\log\left(\frac{1}{\delta}\cdot\log^*|X|\right)$, and let $S\in X^n$ be a database.
Let $\{(X_{d},S_{d},T_{d},\pi_{d})\}_{d=1}^{\log^*|X|}$ denote the outcomes of the execution of \texttt{ConstructPaths} with parameter $t$ on $S$.
Let algorithm \texttt{LevelUp} be executed on $S$ with a parameters $t,d,i$ such that $i$ is an interior point of the database $S_{d+1}$ with depth at least $\Omega\left(\frac{1}{\eps}\log\frac{1}{\delta}\right)$. 
Then, with probability at least $1-\delta$, the algorithm returns an interior point of $S_{d}$ with depth at least $\Omega(t)$. %
\end{lemma}

\begin{proof}[Proof Sketch]
The lemma follows from similar arguments to those given in Section~\ref{sec:TreeLog}. In more details, suppose that $i$ is an interior point of the database $S_{d+1}$, with depth at least $\Omega\left(\frac{1}{\eps}\log\frac{1}{\delta}\right)$. This means that at least $\Omega\left(\frac{1}{\eps}\log\frac{1}{\delta}\right)$ points from $\hat{S}_{d}$ fall off the path $\pi_{d}$ on or before level $i$ of the tree $T_{d}$, and at least $\Omega\left(\frac{1}{\eps}\log\frac{1}{\delta}\right)$ points from $\hat{S}_{d}$ fall off the path $\pi_{d}$ on or after level $i$.

Let $v_{\pi_{d}}^{i}$ denote the node of $\pi_{d}$ at level $i$ of the tree. Since at least $\Omega\left(\frac{1}{\eps}\log\frac{1}{\delta}\right)$ points from $\hat{S}_{d}$ fall off from $\pi_{d}$ on or after level $i$, we get that the weight of $v_{\pi_{d}}^{i}$ is at least $\Omega\left(\frac{1}{\eps}\log\frac{1}{\delta}\right)$. Hence, by the properties of the Choosing Mechanism (Lemma \ref{lem:CM} with $\beta = \delta/2$), the node $v^*$ identified in Step~3 has weight at least $\Omega\left(\frac{1}{\eps}\log\frac{1}{\delta}\right)$ with probability at least $(1-\delta/2)$.
In that case, similarly to the analysis in Section~\ref{sec:TreeLog}, at least one of $\left\{v_{\scriptscriptstyle\rm left}, v_{\scriptscriptstyle\rm right}, v_{\scriptscriptstyle\rm inner\text{-}left}, v_{\scriptscriptstyle\rm inner\text{-}right} \right\}$ is an interior point of $\hat{S}_{d}$, and hence, an interior point of $S_{d}$ with depth at least $\Omega(t)$.
It follows that the Exponential Mechanism (Lemma \ref{prop:exp_mech} with $\beta=\delta/2$) identifies such an interior point of $S_{d}$ with depth at least $\Omega(t)$   with probability at least $(1-\delta/2)$.

\end{proof}

The following lemma specifies the privacy properties of algorithm \texttt{LevelUp}.

\begin{lemma}\label{lem:LevelUpPrivacy}
Fix two neighboring databases $S,S'\in X^n$, and consider running algorithm \texttt{LevelUp} on $S$ and on $S'$ with parameter $d<d^*(S)$. Then for every set of outcomes $T$ we have
$$
\Pr[\texttt{LevelUp}(S)\in T]\leq e^{2\eps}\cdot\Pr[\texttt{LevelUp}(S')\in T]+2\delta.
$$
\end{lemma}

\begin{proof}
Let $S_{d}$ and $S'_{d}$ denote the databases defined in Step~1 of the executions of \texttt{LevelUp} on $S$ and on $S'$, respectively. Since $d<d^*(S)$, by Lemma~\ref{lem:ConstructPathsNeighboring} we have that $S_{d}$ and $S'_{d}$ are neighboring databases. Therefore, the databases $\hat{S}_{d}$ and $\hat{S}'_{d}$ defined in Step~2 of the executions are also neighboring. We then access one of these neighboring databases using the Choosing Mechanism (in Step~3) and using the Exponential Mechanism (in Step~5), and hence, Lemma~\ref{lem:LevelUpPrivacy} follows from the privacy properties of these two private mechanisms (and from composition). %
\end{proof}

\begin{algorithm*}[!htp]

\caption{\bf \texttt{HeavyPaths}}\label{alg:HeavyPaths}

{\bf Input:} Parameters $\eps,\delta$ and a database $S\in X^n$ where $X$ is a totally ordered domain. We assume that $|X|$ is a power of 2.

\begin{enumerate}[leftmargin=15pt,rightmargin=10pt,itemsep=1pt,topsep=1.5pt]

\item Let $t=2k=\frac{2\lambda}{\eps}\cdot\log\left(\frac{1}{\delta}\cdot\log^*|X|\right)$, where $\lambda$ is a global constant. Apply algorithm \texttt{StoppingPoint} on $S$ with parameter $t$ and let $d^*$ denote its output.

\item Run algorithm \texttt{OneRandomPath} on $S$ with parameters
$t$ and $d=d^*-1$. Let $y_{d^*-1}$ denote its outcome.

\item For $d=d^*-2$ down to $1$ do:
\begin{enumerate}[topsep=0pt]
	\item Run algorithm \texttt{LevelUp} on $S$ with parameters $d$ and $t$, and with the point $y_{d+1}$. Denote the outcome as $y_{d}$.
\end{enumerate}

\item Return $y_1$.

\end{enumerate}
\end{algorithm*}

We are now ready to present and analyze algorithm \texttt{HeavyPaths} that runs algorithms \texttt{StoppingPoint}, \texttt{OneRandomPath}, and \texttt{LevelUp} in order to (privately) identify an interior point of its input database. The following lemma specifies the utility guarantees of the algorithm.

\begin{lemma}\label{lem:HeavyPathsUtility}
Let algorithm \texttt{HeavyPaths} be executed on a database $S\in X^n$ of size $$n=\Omega\left(\frac{\log^*|X|}{\eps}\cdot\log\left(\frac{1}{\delta}\cdot\log^*|X|\right)\right).$$ Then the algorithm returns an interior point of $S$ with probability at least $1-O\left(\delta\cdot\log^*|X|\right)$.
\end{lemma}

\begin{proof}
Let $\{(X_{d},S_{d},T_{d},\pi_{d})\}_{d=1}^{\log^*|X|}$ denote the outcomes of the execution of \texttt{ConstructPaths} with parameter $t$ on $S$.
First, by Lemma~\ref{lem:StoppingPointUtility}, with probability at least $(1-\delta)$ we have that the value $d^*$ computed using algorithm \texttt{StoppingPoint} is such that $f_{d^*}(S)\geq\Omega\left(k\right)$. 
Then, by Lemma~\ref{lem:OneRandomPath_Utility}, with probability at least $(1-\delta)$, the point $y_{d^*-1}$ computed by algorithm \texttt{OneRandomPath} is an interior point of $S_{d^*-1}$ with depth at least $\Omega(t)$. %
Therefore, by induction using Lemma~\ref{lem:LevelUpUtility}, with probability at least $(1-\delta\cdot\log^*|X|)$, for every $1\leq d\leq d^*-2$ we have that $y_d$ (computed by algorithm \texttt{LevelUp}) is an interior point of the database $S_{d}$ with depth at least $\Omega(t)$. %
This concludes the proof as $S_1=S$ and therefore $y_1$ is an interior point of $S$.
\end{proof}

We now proceed with the privacy analysis of algorithm \texttt{HeavyPaths}. For two random variables $Z_0,Z_1$ we write $Z_0\approx_{(\eps,\delta)}Z_1$ to mean that for any event $T$ and for any $b\in\{0,1\}$ it holds that $\Pr[Z_b\in T]\leq e^{\eps}\cdot\Pr[Z_{1-b}\in T]+\delta$.

\begin{lemma}\label{lem:HeavyPathsPrivacy}
Algorithm \texttt{HeavyPaths} is $\left(\bar{\eps},\bar{\delta}\right)$-differentially private for
\begin{align*}
\bar{\eps}&=O\left(\eps\cdot\sqrt{\log^*|X|\cdot\log\frac{1}{\delta\cdot\log^*|X|}}+\eps^2\cdot\log^*|X|+\eps\cdot\log n\right),\\
\bar{\delta}&=O\left(\delta\cdot\left(n+\log^*|X|\right)\cdot e^{\bar{\eps}}\right).
\end{align*}
\end{lemma}

\begin{proof}
Let $\BBB_{d}$ denote an algorithm consisting of steps 2 and 3 of algorithm \texttt{HeavyPaths}, with $d$ as the parameter from Step~1.

Fix two neighboring databases $S,S'\in X^n$. For every $d\leq d^*(S)$, by Lemmas~\ref{lem:LevelUpPrivacy} and~\ref{lem:OneRandomPath_Privacy} we have that $\texttt{LevelUp}(S)\approx_{(2\eps,2\delta)}\texttt{LevelUp}(S')$, and that
$\texttt{OneRandomPath}(S)\approx_{(\eps',\delta')}\texttt{OneRandomPath}(S')$ for
$\eps'=6\eps\cdot\log n$ and $\delta'=4\delta n\cdot e^{4\eps\cdot\log n}$. 
Therefore, applying composition theorems for differential privacy (see~\cite{DRV10}; we assume that $\eps\leq1$) to the execution of \texttt{OneRandomPath} and the (at most) $\log^*|X|$ executions of \texttt{LevelUp}, we get that $\BBB_{d}(S)\approx_{(\hat{\eps},\hat{\delta})}\BBB_{d}(S')$, for 
\begin{align*}
\hat{\delta}&=O\left(\delta\cdot\log^*|X|+\delta n\cdot e^{4\eps\log n}\right), \quad\text{ and }\\
\hat{\eps}&=O\left(\eps\cdot\sqrt{\log^*|X|\cdot\log\frac{1}{\delta\cdot\log^*|X|}}+\eps^2\cdot\log^*|X|+\eps\cdot\log n\right).
\end{align*}
Consider the executions of \texttt{HeavyPaths} on $S$ and on $S'$, and let $d^*$ and ${d^*}'$ denote the values obtained in Step~1 of the execution on $S$ and on $S'$, respectively.
By Lemmas~\ref{lem:StoppingPointUtility} and~\ref{lem:StoppingPointPrivacy}, for every set of outcomes $T$ we have that
\begin{align*}
\Pr[\texttt{HeavyPaths}(S)\in T] &\leq
\Pr[d^*>d^*(S)] + \sum_{d\leq d^*(S)}\Pr[d^*=d]\cdot\Pr[\texttt{HeavyPaths}(S)\in T | d^*=d]\\
&=\Pr[d^*>d^*(S)] + \sum_{d\leq d^*(S)}\Pr[d^*=d]\cdot\Pr[\BBB_{d}(S)\in T]\\
&\leq\delta + \sum_{d\leq d^*(S)}\left(e^{\eps}\cdot\Pr[{d^*}'=d]+\delta\right)\cdot\left(e^{\hat{\eps}}\cdot\Pr[\BBB_{d}(S')\in T]+\hat{\delta}\right)\\
&\leq4\hat{\delta}\cdot e^{\hat{\eps}}\cdot\log^*|X| + e^{\hat{\eps}+\eps}\cdot\sum_{d\leq d^*(S)}\Pr[{d^*}'=d]\cdot\Pr[\BBB_{d}(S')\in T]\\
&\leq4\hat{\delta}\cdot e^{\hat{\eps}}\cdot\log^*|X| + e^{\hat{\eps}+\eps}\cdot\Pr[\texttt{HeavyPaths}(S')\in T].
\end{align*}
\end{proof}

Theorem~\ref{thm:mainAdvanced} now follows by combining Lemma~\ref{lem:HeavyPathsUtility} and Lemma~\ref{lem:HeavyPathsPrivacy}.

\section*{Acknowledgements}
We thank Kobbi Nissim and Guy Rothblum for stimulating discussions.

\bibliographystyle{abbrv}

\end{document}